\documentclass{article}

\usepackage{arxiv}
\usepackage{float}

\usepackage{amsfonts}
\usepackage{amsthm}
\RequirePackage{mathrsfs}
\usepackage{natbib}
\usepackage{graphicx}
\RequirePackage{bm} 
\RequirePackage{url}
\usepackage{multirow}
\usepackage{graphicx}

\usepackage{dsfont}
\usepackage{mathtools}
\usepackage{enumerate}
\usepackage{csvsimple}
\usepackage{wrapfig}  %
\usepackage[noabbrev]{cleveref}
\usepackage{xcolor}
\usepackage{algorithm2e}
\usepackage[table]{xcolor}
\usepackage{matern}
\usepackage{comment}
\RequirePackage{snapshot}

\makeatletter
\renewcommand{\normalsize}{%
  \fontsize{12pt}{14pt}\selectfont  %
  \abovedisplayskip      7\p@ \@plus 2\p@ \@minus 5\p@
  \abovedisplayshortskip \z@ \@plus 3\p@
  \belowdisplayskip      \abovedisplayskip
  \belowdisplayshortskip 4\p@ \@plus 3\p@ \@minus 3\p@
}
\makeatother

  \title{Bayesian Repulsive Mixture Modeling with \matern Point Processes}
  \author{Hanxi Sun \And Boqian Zhang \\
  \And Minhyeok Kim \And Vinayak Rao\thanks{varao@purdue.edu (corresponding author)} %
    \\ \\ \hspace{-5in}
    Department of Statistics, Purdue University \\
    \hspace{-5in} West Lafayette, IN 47907, USA}
    \date{}

\begin{document}

  \maketitle

  \begin{abstract}

Mixture models are a standard tool in statistical analyses, widely used for density modeling and model-based clustering. 
In this work, we propose a Bayesian mixture model with repulsion between mixture components.
Such repulsion helps address the problem of overlapping or poorly separated clusters, and assists with model interpretibility and robustness. 
Our modeling approach introduces repulsion via a generalized \matern type-III repulsive point process model, and proceeds by applying a dependent sequential thinning scheme to a latent Poisson point process. 
A key feature of our model is that in contrast to most existing approaches to modeling repulsion, efficient posterior inference is possible via a  Gibbs sampler, one that exploits the latent Poisson of our problem. This novel sampler also allows posterior inference over the number of clusters,  and is of independent interest even in standard clustering applications without repulsion. We demonstrate the utility of the proposed method on a number of synthetic and real-world problems.

\end{abstract}

\noindent%

\section{Introduction} \label{sec:introduction}

Recent advances in statistical and machine learning have placed 
a growing emphasis on  balancing statistical fidelity and predictive accuracy with interpretability, parsimony and fairness.  
In this paper, we focus on interpretability and diversity in mixture modeling applications, through the use of {\em repulsive priors}.
Mixture models are useful both in density modeling applications as well as in clustering applications~\citep{mclachlan1988mixture, banfield1993model, bensmail1997inference}, %
with goals for the latter including data exploration, visualization and summarization. 
For computational tractability, the parameters of the mixture components are typically modeled as independent and identically distributed draws from some base-distribution. 
However, unless the clusters are widely separated, this can result in multiple overlapping clusters, %
leading to redundancy, and lack of interpretability.  %
Furthermore, since mixture models are typically composed of simple parametric components, even if the data exhibits clear clustered structure, slight deviations of individual clusters from the parametric form will again result in overlapping and inconsistent number of components~\citep{beraha25}.

A recent and popular approach addresses this problem by jointly sampling all component parameters from a \emph{repulsive prior} that %
penalizes configurations with components situated too close to each other. 
Such priors typically draw from the point process literature, examples including Gibbs point processes~\citep{stoyan1987stochastic} and determinantal point processes~\citep{hough2006determinantal, lavancier2015determinantal}.
Mixture models with repulsion have been shown to provide simpler, clearer and more interpretable results, often without too much loss of predictive performance~\citep{petralia2012repulMix, xu2016bayesian, bianchini2018determinantal, beraha25}.  
Nevertheless, they present computational challenges, often involving intractable normalization constants or reversible-jump algorithms.

In this work, we propose a new, flexible class of repulsive priors based on the \matern type-III point process~\citet{matern1960, matern1986}. %
An advantage of these is the ability to flexibly introduce new, mechanistic repulsive mechanisms, as shown  recently in~\citet{rao2017matern}. That work also developed an efficient Markov chain Monte Carlo (MCMC) algorithm for posterior sampling. %
We bring this process to the setting of mixture models, using them as a repulsive prior over the number of components and their locations.
Treating the \matern realization as a latent, rather than a fully observed point process raises computational challenges that the algorithm from~\citet{rao2017matern} does not handle. 
We develop an efficient MCMC sampler for our model and demonstrate the practicality and flexibility of our proposed repulsive mixture model on a variety of datasets.
Our sampler is also useful to sample the number of components %
in mixture models without repulsion, as an alternative to often hard-to-tune reversible jump MCMC methods~\citep{richardson1997bayesian}. 

We organize this paper as follows. \Cref{sec:background} reviews the generalized \matern type-III point process, while \Cref{sec:model} and \Cref{sec:method} outline our proposed {\em \matern Repulsive Mixture Model} (MRMM) and our novel MCMC algorithm. \Cref{sec:related-work} discussed related work on repulsive mixture models, and we apply our model to a number of datasets in~\Cref{sec:experiment}.

\section{\matern repulsive point processes} \label{sec:background}

The Poisson process~\citep{kingman1992poisson} is a {\em completely random} point process, where events in disjoint sets are independent of each other.
To incorporate repulsion between events,~\citet{matern1960, matern1986} introduced three spatial point process models that build on the Poisson process.
The three models, called the \matern hardcore point process of type I, II and III, only allow point process realizations with pairs of events separated by at least some fixed distance $\eta$, where $\eta$ is a parameter of the model. 
The three models are constructed by applying different thinning or event-deletion schemes on a {primary} homogeneous Poisson point process.
Despite being theoretically more challenging than the other two processes, the type-III process has the most natural thinning mechanism, and supports higher densities of points.
~\citet{rao2017matern} showed how this can easily be generalized to include probabilistic thinning and spatial inhomogeneity. 
Furthermore,~\citet{rao2017matern} showed that posterior inference for a completely observed type-III process can be carried out in a relatively straightforward manner.
For these reasons, we will focus on the generalized \matern type-III process, and for simplicity, will refer to this simply as the \matern process in the rest of this paper. 

Formally, the \matern process is a finite point process defined on a space $\Location$, parameterized by a thinning kernel $\Kernel_\thin:\Location\times\Location\to[0, 1]$ and a nonnegative intensity function $\Rate_\Theta:\Location \to [0,\infty)$. 
We decompose $\Rate_\Theta(\theta)$ as $\Rate_\Theta(\theta)=\mRate\cdot p_\Theta(\theta)$, for a finite {normalizing}
constant $\mRate > 0 $ and some probability density $p_\Theta(\theta)$ on $\Theta$.
Simulating this process proceeds in four steps:
\begin{enumerate}
    \item Simulate the primary Poisson process $\ssF=\cbr{\s_1,\dotsc,\s_{\card{\ssF}}} \subset \Theta$ %
    with intensity $\Rate_\Theta(\cdot)$. 
\item Assign each event $\s_j$ in $\ssF$ an independent {\em birth-time} uniformly on $\Time=[0, 1]$. 
\item Sequentially visit events in $\ssF$ according to their birth-times from the oldest to the youngest and attempt to thin/delete them.  
At step $j$, the $j$th oldest event $(\theta,t)$ is thinned by each surviving older primary event $(\theta', t'), t' < t$ with probability $\Kernel_\thin(\theta,\theta')$. 
\item Write $\ss$ and $\ssGt$ for the elements of $\ssF$ that survive and are thinned from the previous step, respectively. The set $\ss$ forms the \matern process realization.
\end{enumerate}

Different choices of the thinning kernel $\kernel\s{\s_j}$ give different variants of the \matern process. For a hardcore \matern process (\Cref{fig:matern}), $\kernel\s{\s_j} = \ind{\|\s-\s_j\| < \thin}$,  %
so that thinning is deterministic: all newer events within radius $\thin$ of a previously survived event are thinned.
Other approaches are probabilistic thinning~\citep{rao2017matern}, where $\kernel\s{\s_j} = \thin_p\ind{\|\s-\s_j\| < \thin_R}$ (with $\thin_p \in [0, 1]$), or the smoother squared-exponential thinning, where
$\kernel\s{\s_j} = \exp(-\frac{\|\s-\s_j\|^2}{2\thin}) $. 
~\citet{huber2009likelihood} propose soft-core thinning, where each event $\s_j$ has its own thinning radius $\thin_j$ drawn from some distribution, and  $\kernel\s{\s_j} = \ind{\|\s-\s_j\| < \thin_j}$.
\begin{figure}
\centering
	\includegraphics[width=0.8\linewidth]{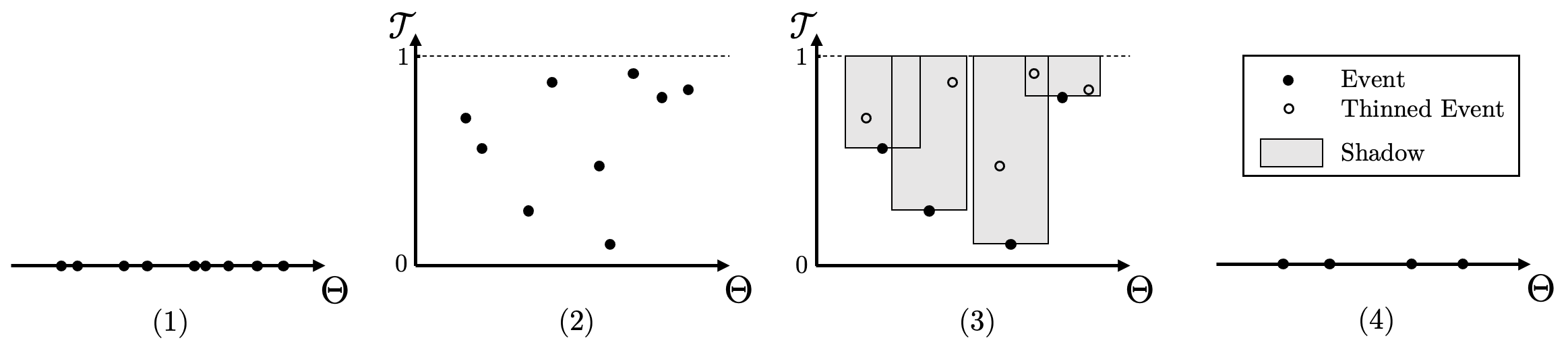}
	\caption{
		The %
        generative process of a one-dimensional hardcore \matern process. %
        }
	\label{fig:matern}
\end{figure}

Observe that %
the set $\cbr{(\s_1,t_1),\dotsc,(\s_{\card{\ssF}}, t_{\card{\ssF}})}$ itself forms a Poisson process on $\Location \times \Time$, with intensity $\rate{\s,t} = \Rate_\Theta(\s)\ind{[0,1]}(t)$. 
We write $\F_{\Location \times \Time}$ for this extended process, and $\ttF = \proj{\Time}{(\F_{\Location \times \Time})}$ for the set of birth-times.
We use $\G_{\Location \times \Time}$ for the extended \matern events, $\tt$ for the associated birth-times, and $\Gt_{\Location \times \Time}$ and $\ttGt$ for their thinned counterparts.

Following~\citet{rao2017matern}, we define %
a {\em shadow function} $\Shadow_\thin: \Location\times\Time\to[0, 1]$. %
This gives the probability that an event $(\s^*, \t^*)\in\Location\times\Time$ is thinned by a collection of events $\F_{\Location\times\Time}$ as 
\begin{align}
	\shadow{(\s^*,\t^*)}{\F_{\Location\times\Time}}       & =  1 - \prod_{(\s,\t)\in\F_{\Location\times\Time}}\left({1-  \ind{(\t,1]}(\t^*) \kernel{\s^*}{\s}}\right).
\end{align}
Above, the $\ind{(t,1]}(t^*)$ term reflects that an event $(\s^*,\t^*)$ can only be thinned by earlier events. 
We write $\maternThin{\F_{\Location\times\Time}, \thin}$ for the sequential thinning process that assigns elements of $\F_{\Location\times\Time}$ to one of $\G_{\Location\times\Time}$ or $\Gt_{\Location\times\Time}$ according to kernel $\Kernel_\thin$, and  $\proj{\Theta}(G_{\Location\times\Time})$ for the projection of events in $G_{\Location\times\Time}$ onto $\Theta$.
The generative process of $\ss\sim\maternP{\Rate, \thin}$ is then %
\begin{eqnarray*}
  \given{\F_{\Location\times\Time}}\Rate & \sim & \poisP{\rate{\cdot,\cdot}}, \\    \given{\G_{\Location\times\Time}, \Gt_{\Location\times\Time}}{\F_{\Location\times\Time}, \Kernel_\thin} & \sim & \maternThin{\F_{\Location\times\Time}, \thin}, \quad %
    \ss  =  \proj{\Theta}(\G_{\Location\times\Time}).
\end{eqnarray*}

\section{\matern repulsive mixture model (MRMM)}\label{sec:model}

To extend the above to a repulsive mixture model, %
we treat $\Location$ as a parameter-space, and introduce weight-space $\Weight = [0,\infty)$. For some density $p_\Weight(\cdot)$ on $\Weight$, we now consider 
 a primary process $F$ on $\Location\times\Weight\times\Time$ with $\Time=[0,1], \Weight=[0,\infty)$. Unlike before, we model $\F$ as a Poisson process conditioned to have at least one event, with intensity function equal to 
\begin{align}
	\label{eqn:rate}
	\rate{\s, \w, \t} = \mRate\cdot\ps\s\cdot \pw\w\cdot \ind{[0,1]}(t).
\end{align}
We set $\pw\w = \dGamma{\w\; \aw, 1}$, with $\ps\s$ a problem-specific prior over component parameters.
Given $F$, we produce a \matern realization $\G=\cbr{(\s_1,\w_1,\t_1),\dotsc,(\s_{|G|},\w_{|G|},\t_{|G|})}$ by applying $\maternThin{F,\thin}$ for some kernel $\Kernel$ on $\Theta$ with parameter $\thin$.
Each element $(\s,\w,\t) \in G$ will form a component of a mixture model, with  $\s$ and $\w$ representing the parameter and unnormalized weight of that component. %
Our model thus serves as a prior over both the number of components in a mixture model, as well as the component weights and locations.
Since $F$ is defined to have at least 1 event, and since events in $F$ can only be thinned by surviving events, the resulting mixture model will have at least one component.

For a set $A$, write $\sum A$ for the sum of its elements. %
It is well known~\citep{devroye1986non} that 
after normalizing the Gamma-distributed $\ww$, the mixture probability $\frac{\G_\Weight}{\sum G_{\Weight}} := \{\frac{\w_1}{\sum {G_{\Weight}}} ,\dotsc,\frac{\w_1}{\sum {G_{\Weight}}}\}$ follows a symmetric Dirichlet$(\alpha)$ distribution.
Let $\px\cdot{\s}$ be some family of probability densities  parameterized by $\s\in\Location$; this will correspond to the mixture components. Then given $G$, we model the observed data $\XX = \cbr{x_i, i=1, \dots, n}$ as follows:
\begin{align}
	\label{eqn:mixture}
	\given{x_i}{\G} 
     \overset{\mathrm{iid}}{\sim} \sum_{(\s,\w,\t) \in G} \frac{\w}{\sum {G_{\Weight}}}\px\cdot{\s}, \quad\quad  i=1,\dotsc, n.
\end{align}

As an example, if the observations lie on a Euclidean space, $\px\cdot{\s}$ could be a normal distribution, with $\s$ representing the location and variance of a component in a Gaussian mixture model. 
In this case, the density $\ps\s$ might be a Normal-Inverse-Wishart distribution.

If the thinning kernel $\Kernel_\thin$ equals $0$, our model reduces to a standard mixture model, with i.i.d.\ component parameters, Dirichlet-distributed component weights, and a conditional Poisson distribution on the number of components.
Different settings of $\Kernel_\thin$, whether hardcore, probabilistic or squared-exponential thinning, allow different kinds of repulsion between the component parameters. 
In this work, we only place repulsion between the component parameters $\theta$ and not the component weights $w$.
Further, in many settings we allow $\Kernel_\thin$ to only depend on a subset of the components of $\theta$. 
For instance, writing $\s = (\s^\mu, \s^\sigma)$ where $\s^\mu$ is the component location and $\s^\sigma$ is the component variance, it is common to enforce repulsion only between the component locations, but not their variances.
This can easily be achieved by setting $\Kernel_\thin$ to depend only on $\s^\mu$.

To complete the Bayesian model, we specify hyperpriors on $\mRate$ and $\thin$, as well as on any hyperparameters of $p_\Theta(\s)$. 
The last is problem-specific, and is no different from models without repulsion. 
A natural prior for $\mRate$ is the Gamma distribution. %
For the hardcore process, where $\thin$ is the thinning radius, or for the squared-exponential thinning kernel, where $\thin$ is the lengthscale parameter, we can use a Gamma hyperprior. 
For probabilistic thinning, where $\thin=(R, p)$, we can use a Beta prior on the thinning probability $p$, and a Gamma prior on the thinning radius $R$. 
We include further discussion of the parameters of the thinning kernel in %
\Cref{sec:experiment} and the supplementary material. 

Write $\zz=\rbr{\z_1, \dotsc, \z_n}$ for the cluster assignments of the data in \cref{eqn:mixture}, with $\z_i\in\cbr{1, \dots, |G|}$. %
With hyperpriors omitted for simplicity, the generative process of MRMM is %
\begin{align}
	\label{eqn:mrmm}
		\given{\F}{\Rate} & 
			\sim \poisP{\rate\cdot}\big|{\card{F} > 0}, \ \ \quad \qquad
		\given{\G, \Gt}{\F, \Kernel_\thin}  
			\sim \maternThin{\F, \thin}, \\
		\given{\z_i}{G} & 
         \overset{\mathrm{iid}}{\sim} \multinomial{\frac{\G_\Weight}{\sum\G_\Weight}},\quad
        \given{x_i}{\z_i, G}  
        \sim \px\cdot{\s_{\z_i}}, \qquad  i = 1,\dotsc, n. \nonumber
\end{align}
Write $\mathcal{M} = \{\textsc{thin},\textsc{no-thin}\}$ for a two point `mark' space. The proposition below gives the joint density of all variables, and is useful for deriving our posterior sampling algorithm.
\begin{theorem}
	\label{prop:X-G-Gt}
    Write $\scP_\lambda$ for the law of a rate-$\lambda(\cdot)$ Poisson process on $\Space\times \mathcal{M}$. Then the measure of the tuple $\XX$, $\G$, $\Gt$ has density with respect to ${d}x^n\times\scP_\lambda$ given by
	\begin{align}
      \pgiven{\XX, \G, \Gt}{\Rate, \thin} & =
      \frac{\ind{}(|\G \cup \Gt| > 0)}{1-e^{\int_{\Space}- \lambda(\s, \w, \t)\dif\s \dif\w \dif\t }} \nonumber \\
    & \quad  \prod_{\g\in\G}\sbr{1-\shadow\g\G}\prod_{\gt\in\Gt}\shadow\gt\G  
		 \prod_{i=1}^{n}\sum_{(\s, \w, \t)\in\G}{\frac\w{\sum\G_\Weight}\px{x_i}{\s} }.
	\end{align}	
\end{theorem}

\section{Posterior inference for MRMM}\label{sec:method}

Given a dataset $\XX=\{x_1,\dotsc,x_n\}$ modeled with MRMM, the posterior distribution $\pgiven{G, \zz, \mRate, \thin}{\XX}$ summarizes information about the component weights and locations (through $G$), and the cluster assignments (through $\zz$). 
We construct a Markov chain Monte Carlo (MCMC) sampler to simulate from this.
Our sampler also imputes the thinned events $\Gt$,
and proceeds by sequentially updating $\mRate, \thin, \G, \Gt$ and $\zz$ according to their conditional posterior distributions. 
Given the pair $(\G,\Gt)$, updating the remaining variables is fairly straightforward, and we show how the latent Poisson structure makes updating these variables relatively easy too.
Below, we present full details of the Gibbs steps.

\noindent\textbf{1) Updating thinned events $\Gt$:}
From the data generation process, it follows that given $\G$,  $\Gt$ is independent of $\zz$ and $\XX$.
Furthermore, for \matern type-III processes, events in $\Gt$ can only be thinned by events in $\G$, suggesting that given $\G,\mRate,\thin$, the events in $\Gt$ do not interact with each other, and form a Poisson process.
The result below formalizes this:
\begin{prop}
\label{prop:Gt}
Given all other variables, the conditional distribution of the thinned events $\Gt$ is a Poisson process with intensity
${\rate\cdot \shadow\cdot{G}}$.
\end{prop}
\noindent This result follows \citet{rao2017matern}, though in this work, we are conditioning on $\G \cup \Gt$ having at least one event. Our proof, included in the supplement, is also simpler and cleaner, exploiting~\Cref{prop:X-G-Gt} and working with densities with respect to the rate-$\lambda$ Poisson measure. 
Since ${\shadow\cdot{G}} \le 1$, we can easily use Poisson thinning~\citep{lewis1979simulation} to simulate this Poisson process: simulate a Poisson process with intensity $\Rate(\cdot)$ on the whole space $\Space$, and then keep each event $\gt$ in it with probability $\shadow\gt\G$. 
This makes jointly updating the entire set $\Gt$ easy and efficient, without any tuning parameters.

\noindent\textbf{2) Updating the \matern events $G$:}
This step is more challenging, since %
the \matern events interact with each other, and with the clustering structure of the data.
Instead of trying to independently update the entire $G$, we do so one component at a time. %

We first %
discard the cluster assignments $\zz$, these are easily resampled in step 3 below. 
We then make a pass through the elements of $\G \cup \Gt$, using~\Cref{prop:X-G-Gt} to reassign each to either $\G$ or $\Gt$. %
At the end of this, we have an updated pair $(\G^*,\Gt^*)$. %
While the union $G \cup \Gt$ is unchanged, our ability to efficiently update $\Gt$ in the previous step suggests fast mixing.

In our experiments however, we sometimes observed poor mixing, especially with hardcore thinning.
The deterministic thinning of this  process forbids elements of $\G^*$ from lying within each others' shadow, and also requires $\Gt^*$ to lie in the shadow of $\G^*$, making it hard to switch an event from the \matern set to thinned set, or vice versa. 
In settings where $\Gt$ has few events, this chain will mix poorly, and when there is no repulsion (so that $|\Gt|=0$), this Markov chain is no longer ergodic.
To address this, at the start of this step, we augment our MCMC state-space with an independent rate-$\au\rate\cdot$ Poisson process $\Ft\subset\Space$: %
\begin{align}
	\given{\Ft}{\au,\Rate}\sim\poisP{\au\rate\cdot}.
\end{align}
We call $\au>0$ the augmentation factor, which forms a parameter of our MCMC algorithm.
Having simulated $\Ft$, we cycle through the elements of $\G \cup \Gt \cup \Ft$, sequentially relabeling each event as `survived', `thinned' or `augmented' to produce a new triplet $\G^* \cup \Gt^* \cup \Ft^*$. 
This relabeling is carried out to preserve the joint conditional of $\G^*, \Gt^*, \Ft^*$, and after discarding $\Ft^*$, we have updated $(\G,\Gt)$ while maintaining their conditional distribution.

Since $\Ft$ is independent of everything else, it more easily allows events to be introduced into, and removed from $\G$. %
Each relabeling step is straightforward, and requires computing a three-component probability.
For each $e \in \G \cup \Gt \cup \Ft$, write $\Ge, \Gte$ and $\Fte$ for the sets resulting from removing $e$ (only one of these will change). 
Write $\Se$ for the sum of the weights after removing $e$: $\Se = \sum {\proj{\Weight}(\Ge)}$.
For any  $x_i \in \XX$ and event $g = (\theta,w,t) \in G$, write $l^g_{i} = \w\px{x_i}{\s}$, and $\Le_i = \sum_{g \in \Ge}l^g_i$: %
this is the unnormalized likelihood of observation $i$ with event $e$ taken out, and with its cluster assignment marginalized out. 
Then, following \Cref{prop:X-G-Gt}, %
the probabilities of  ``survived'', ``thinned'' or ``augmented'' are
\begin{align}
	\label{eqn:relabel-posts}
    P(e \in G| -)  & \propto  
    	\prod_{i=1}^n \frac{l^e_i\emp{+\Le_i}}{\Se\emp{+\proj{\Weight}(e)}}  
		\prod_{\g\in\Ge\emp{\cup\{e\}}}\sbr{1-\shadowA\g{\Ge\emp{\cup\{e\}}}}  \prod_{\gt\in\Gt}\shadowA\gt{\Ge\emp{\cup\{e\}}}, \nonumber \\
    P(e \in \Gt| -) & \propto 
    	\prod_{i=1}^n \frac{\Le_i}{\Se}  
    	\prod_{\g\in\Ge}\sbr{1-\shadowA\g{\Ge}} 
    	\prod_{\gt\in\Gte\emp{\cup\{e\}}}\shadowA\gt\Ge, \\
    P(e \in \Ft| -) & \propto \emp{\au}
    	\prod_{i=1}^n \frac{\Le_i}{\Se} 
    	\prod_{\g\in\Ge}\sbr{1-\shadowA\g{\Ge}}
    	\prod_{\gt\in\Gte}\shadowA\gt\Ge. \nonumber 
\end{align}
Having cycled through all elements of $\G \cup \Gt \cup \Ft$, we have a new partition $(\G^*,\Gt^*,\Ft^*)$, after which the augmented Poisson events $\Ft^*$ are discarded. %
The augmented factor $\au$ in this procedure governs the cardinality of augmented events $\Ft$. 
A larger $\au$ results in faster mixing, but higher computational cost.
Our experiments in the supplementary material suggest that a moderate augmentation factor of 5 adequately balances mixing and computation. 

\noindent\textbf{3) Updating cluster assignments $z$ and component weights $\ww$:} 
Given $\XX$ and mixture parameters $\ss$ and $\ww$, we can easily resample the assignments $\zz$ that were discarded at the start of the previous step. 
This is no different from standard mixture models; for observation $i$:
$
  p(z_i = g | - ) \propto l^g_i, \ \forall g \in G. 
$
Clusters assignments for all observations are conditionally independent, so that these assignments can be carried out in parallel.

In light of the first two update steps, updating the weights $\ww$ is not strictly necessary, nevertheless it is very straightforward and improves mixing.
Given cluster assignments $\zz$ and the number of mixture components $\card{\G}$, the mixture weights $\ww=\cbr{\w_j, j=1,\dots,|\G|}$ are independent of the other variables.
A priori, the $\w_j$'s are independent $\dGamma{\aw, 1}$ random variables, or equivalently, are obtained by multiplying a sample from a $\dirichlet(\alpha,\dotsc,\alpha)$ distribution (the normalized weights) with an independent sample from a $\dGamma{|G|\alpha, 1}$ distribution (the sum of the weights)~\citep{devroye1986non}.
We work with the latter representation, and seek to simulate from the posterior distribution of the normalized weights and the sum of the weights.
It is easy to see that these continue to be independent under the posterior.
The sum of the weights plays no role in the likelihood, and continues to follow a $\dGamma{|G|\alpha,1}$ distribution, while the Dirichlet-multinomial conjugacy implies that 
the normalized weights follow a $\dirichlet(\aw+n_1, \dots, \aw+n_{\CG})$, with $n_j$ the number of observations in component $j$.

\noindent\textbf{4) Updating component locations $\ss$ and \matern birth-times $\tt$:} 
Again, updating $\ss$ and $\tt$ is not strictly necessary, nevertheless, we find this improves mixing.
With $|\G|$ and $|\Gt|$ determined, updating these is straightforward, if a little tedious.
Unlike standard mixture models, because of repulsion, component locations are not conditionally independent.
Write $\s_j$ for the location of $j$-th component, and $\G_{(\s_j=\s^*)}$ for $G$ with $\s_j$ updated to $\s^*$. Then, writing $\XX_j$ for the observations assigned component $j$, the conditional of $\s_j$ is
\begin{align}
	\label{eqn:ss}
	\pgiven{\s_j=\s^*}{-} %
	& \propto %
    \ps{\s^*}\prod_{x\in\XX_j}\px x{\s^*}
    \prod_{\g\in\G}\bsbr{1-\shadow\g{\G_{(\s_j=\s^*)}}}\prod_{\gt\in\Gt}\shadow\gt{\G_{(\s_j=\s^*)}}.
\end{align}
The last two products account for how changing the $j$th event's location changes the shadow, and therefore the probability of the current \matern and thinned events.
The other two terms are the prior and likelihood of $\s_j$ under a mixture model without repulsion.
A simple way to simulate from this is with a Metropolis-Hastings step,
and when the prior $p_\s$ is conjugate to the likelihood $p(x\vbar\s)$, a natural choice for the proposal distribution 
is the posterior distribution if there were no repulsion:
$	q_j\rbr{\s_j} \propto \ps{\s_j} \prod_{x\in\XX_j}\px{x}{\s_j}$.

Like the component locations, the birth-times $\tt$ of the \matern events can also be updated one at a time.
Given the component locations, $\tt$ is independent of the observations or their cluster assignments, and one only needs to consider their impact on the shadow (\Cref{prop:X-G-Gt}).
Specifically, if $t_j$ is the birth time of the $j$-th event, then
\begin{equation*}
	\pgiven{\t_j}{-} 
	\propto p\big({\G, \Gt}\,\big|\,{\Rate, \thin}\big) %
	\propto \prod_{\g\in\G}\bsbr{1-\shadow\g\G}\prod_{\gt\in\Gt}\shadow\gt\G.
\end{equation*}
Since $t_j \in [0,1]$, simulating from this is straightforward, though we can simplify this further.
When the thinning kernel is symmetric, for two events $j$ and $k$, the probability of $k$ thinning $j$ if $k$ were older, is the same as $j$ thinning $k$ if $j$ were older. Thus, changing $t_j$ will only change which thins which, and not affect the thinning probability, so that %
the first product term can be dropped.
Next, the birth-times of the thinned events $\gt_j=(\st_j, \wt_j, \tT_j)\in \Gt$ can be used to partition the interval $\Time=[0,1]$ into segments $[\tT_j, \tT_{j+1})$, $j=1,\dotsc,\CGt-1$. 
If the thinning probability is a function only of separation in space (as is the case with all kernels we have considered), then the probability of $t_j$ within each segment is constant, depending only on the identities of the thinned events born before and after the interval $[\tT_j, \tT_{j+1})$. 
For any time $t$, define $\Gt^{\le t}$ as the  events in $\Gt$ born before or at $t$, and define $\Gt^{>t}$ similarly. Then
\begin{align}
	\label{eqn:tt}
    \pgiven{\t_j\in[\tT_{j}, \tT_{j+1})}{-} & \propto \prod_{\gt \in \Gt^{\le t_j}}\shadow{\gt}{\Gnj}
		          \prod_{\gt \in \Gt^{> t_j}}\shadow{\gt}{\G}.
\end{align}
Having picked a segment, the exact value of $t_j$ is drawn uniformly within the segment.

\noindent\textbf{5) Updating hyperparameters:} 
Hyperparameters include the primary Poisson process intensity, and those in the thinning kernel.
The intensity $\mRate$ controls the cardinality of $F$, and it is easy to show that with a $\dGamma{a,b}$ prior, and with the constraint $\card\F > 0$, the conditional posterior is
$	\pgiven{\mRate}{-} \propto\frac1{1-e^{-\mRate}}\dGamma{\mRate\; a+\CF, b+1}$.
Write $\nu$ for any parameters of the normalized Poisson intensity $p_\Theta(\theta\vbar\nu) = \lambda_\Theta(\theta)/\mRate$.
For a prior $p_\nu(\nu)$, the conditional simplifies as
$ p(\nu\vbar-) \propto p_\nu(\nu) \prod_{\theta \in F_\Theta} p_\Theta(\theta\vbar\nu) $.
Finally, writing $\pthin$ for the prior for the thinning parameter $\thin$, the posterior is
$ \pgiven{\thin}{-}\propto\pthin(\thin)\prod_{\g\in\G}\sbr{1-\shadow\g\G}\prod_{\gt\in\Gt}\shadow\gt\G.  $
All three distributions above can be updated using any standard MCMC kernel.

\section{Related Work} \label{sec:related-work}
Work on repulsive mixture models dates back to at least \citet{dasgupta1999learning}, who %
demonstrated the importance of separated components for learning mixture models. %
An early Bayesian mixture model with repulsion was proposed in~\cite{petralia2012repulMix}. 
Here, repulsion was induced through a Gibbs point process mechanism:  
specifically, the prior probability of any configuration of component locations was proportional to the product of individual component probabilities multiplied by a term that penalizes nearby components.
The authors there considered two types of penalties, one corresponding to a product of penalty terms for each pair of components, and one depending on the minimum separation between components. %
\citet{xie2019bayesian} and \citet{quinlan2018density} generalized this model slightly, and also derived posterior rates of convergence.
\citet{fuquene2019choosing} considered a similar approach to~\citet{petralia2012repulMix}, though they framed their work in the more general setting of {\em non-local priors}. 
Here, given a collection of nested models, parameter configurations in a more complex model that result in an identical density to some configuration in a simpler model are given zero probability.
All these works however face computational challenges: the flexibility Gibbs processes comes at the cost of intractable normalization constants. %
This is especially severe when trying to infer parameters of the repulsive penalty, or switch between models with different numbers of components.
Our work replaces the Gibbs point process with the \matern type-III process, though %
one can use other underdispersed point processes. 
In~\cite{bianchini2018determinantal}, the authors use a determinantal point processes (DPP) \citep{hough2006determinantal, scardicchio2009statistical, lavancier2015determinantal}. %
While mathematically and computationally elegant, DPPs are not as mechanistic and directly interpretable as our thinning mechanism.
In our experiments, we compare with the models of~\citet{xie2019bayesian} and~\citet{bianchini2018determinantal}.
More recently, in~\citet{beraha2022mcmc}, the authors propose an exact MCMC sampler that like ours work side-steps the need for reversible jump MCMC.
This auxiliary variable method focuses on a different class of repulsive models, and relies on a perfect simulation algorithm. Finally, \citet{beraha25} propose a general framework that subsumes both repulsion and attraction in mixture models, though their focus is on the characterization of (rather than efficient simulation from) conditional distributions.

We end by noting that another line of work takes a post-processing approach, deliberately using mixtures with a large number of components, and then discarding unoccupied clusters \citep{fruhwirth2019here, saraiva2020bayesian}, and merging nearby clusters together \citep{malsiner2016model}.
Unlike model-based approaches like ours, these are a bit ad hoc, making it difficult to coherently calibrate uncertainty, especially in more complicated hierarchical models.
We refer the reader to~\citet{fruhwirth2019here} for a comprehensive overview of these and related issues.

\section{Experiments} \label{sec:experiment}

In this section we evaluate different settings of our MRMM model and MCMC algorithm, and compare with two other repulsive models: the DPP-based method of \citet{bianchini2018determinantal} and the repulsive Gaussian mixture model of \citet{xie2019bayesian}. 
We implemented our method as a Python3 package \texttt{mrmm}\footnote{available in supplementary material}.
An \texttt{R} implementation of the method of \citet{bianchini2018determinantal} was acquired directly from the authors, while a \texttt{MATLAB} implementation of the method of \citet{xie2019bayesian} was obtained from their supplementary material.

For our model, we placed a $\dGamma{1, 1}$ prior on the unnormalized weights $\ww$ %
and a $\dGamma{1,0.1}$ prior on the primary process intensity $\mRate$.
We considered three thinning kernels, the hardcore, probabilistic and squared-exponential kernel. 
The supplementary material includes more details of the experimental setup.
Typically, we ran 5000 MCMC  iterations, with the first half discarded as burn-in.
To evaluate sampler efficiency, we first computed the effective sample size (ESS) of a number of posterior statistics (we report this  only for $C$, the number of components, though others like the parameter $\eta$ perform similarly).
Dividing this by the total sampler runtime gives the ESS per second (ESS/s), an estimate of the number of independent samples produced per second. 
Since all models were implemented in different languages, this metric should be viewed not as an exact measure of performance, but rather to understand mixing, and how they scale. Ultimately though, we believe the biggest advantage of our sampler is its simplicity.

We also evaluated the different models using statistical performance and parsimony. For the former, 
we reported the predictive likelihood $\lnpTest$ of a held-out test dataset $\XX_{test}$, as well as 
the log pseudo-marginal likelihood $\text{LPML} =\sum_i\log p\rgiven{x_i}{\XX^{-i}}$ where $\XX^{-i}$ denotes the dataset without the $i$-th observation \citep[see][]{bianchini2018determinantal}. 
To assess the parsimony, we reported the posterior mean and variance of the number of components ($\EC$ and $\VarC$), as well as a central estimate of the posterior clustering structure (a `median' posterior clustering).
The latter was obtained by minimizing the posterior expectation of Binder’s loss function under equal misclassification costs \citep{bianchini2018determinantal, lau2007bayesian}. 
We denote the number of components in this estimate as $\hC$.

\subsection{Study of thinning kernels and thinning strengths} \label{sec:exp-syn_mn}
We first study the effect of different thinning kernels and thinning strengths on MRMM inferences. 
\Cref{tbl:kernels_mn} lists thinning kernels and parameters used. %
\begin{table*}
	\centering
    \footnotesize
	\begin{tabular}{|l|ll|l|}
		\hline
		Thinning Kernel     & \multicolumn{2}{l|}{Thinning Parameter}   & Expression                                                \\\hline\hline
		Hardcore            & $\thin=R$       & Radius $R>0$            & $\kernelR\s{\s'}=\ind{\|\s-\s'\|<R}$                      \\\hline      
		Probabilistic       & $\thin=(R, p)$  & Radius $R>0$, Probability $p\in[0,1]$            & $\kernelRp\s{\s'}=p\ind{\|\s-\s'\|<R}$                    \\
		Squared-exponential & $\thin=l$       & Lengthscale $l>0$       & $\kernelSp\s{\s'}l=\exp\cbr{-\frac{\|\s-\s_j\|^2}{2 l}} $ \\\hline
	\end{tabular}
	\caption{Thinning kernels used in experiments}
	\label{tbl:kernels_mn}
\end{table*}
We consider a series of two-dimensional Gaussian mixture models, %
each with four equally weighted, unit-variance Gaussian components, located at $(-d/2, 3d/2)$, $(d/2, d)$, $(d, -d)$, $(-3d/2, -3d/2)$. %
Training and test datasets of size 200 and 100 were simulated for $d=1,2,3,4$. 
We set $p_\Theta(\s)$ to a Gaussian with mean zero and covariance $10 I_2$, and  placed an inverse-Wishart prior with two degrees of freedom and a scale matrix $I_2$ on the covariances. 
When learning the thinning radius $R$ or lengthscale $l$, we placed a $\dGamma{4, 2}$ prior with mean 2 and variance 1.

\begin{figure}
	\centering
	\includegraphics[width=\linewidth]{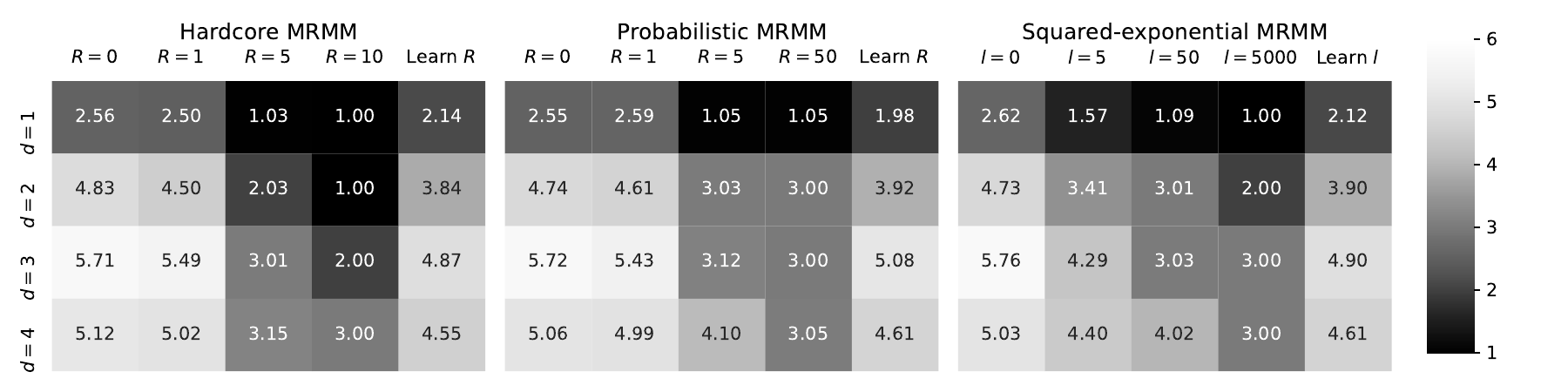}
	\centering
 \includegraphics[width=\linewidth]{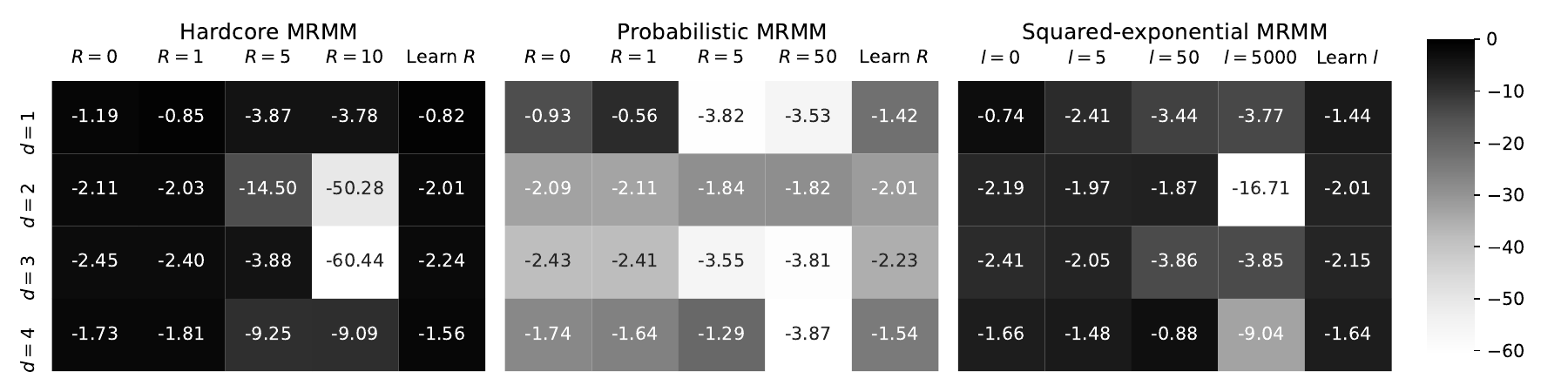}
	\caption{
		(Top) Posterior mean of number of clusters $\EC$, (Bottom) Difference between test likelihood under the posterior and the true model $M_0$, $\lnpTest - \ln \pgiven{\XX_{\text{test}}}{M_0}$.
	}
	\label{fig:syn-sum-lnpTest-diff_mn}
\end{figure}

The supplementary material discusses the results in detail, we focus here on 
\Cref{fig:syn-sum-lnpTest-diff_mn}, whose top and bottom panels evaluate parsimony and the goodness-of-fit.
As expected, increasing repulsion strength results in greater parsimony, with the posterior mean of the number of clusters dropping. 
Interestingly, moderate values of repulsion do not significantly harm the model fit.  
However, a strong repulsion strength does result in a drop in predictive power, especially for the hardcore MRMM.
This is a pattern we will continue to see with the real data.
In the setting where we learn $R$, we observe good predictive performance, and reasonable parsimony, though a few settings suggest that a stronger prior might be needed.

\subsection{Setting thinning parameters via empirical Bayes}
\label{sec:eb}
There are a few ways to set the parameters of the \matern kernel. 
The first is simply by calibration through repeated prior simulation: unlike Gibbs-type repulsive priors, simulation under our model is easy and efficient. 
Another approach is the prior elicitation method from \cite{beraha2022mcmc}. %
Denote by $\pi(r)$ the kernel density estimate of the pairwise distances between observations, and fix $\rloc$ as the smallest local minimum of $\pi(r)$:\\
$
\hspace*{.8in} \rloc = \min_{r>0}\{r : r\text{ is a local minimum for }\pi(r)\}.
$\\
The rationale here is that with a multimodal density $\pi(r)$,  the smallest group of pairwise-distances reflects within-cluster distances, with the rest largely corresponding between-cluster distances. Thus, the smallest local minimum of the density typically lies between the mode of within-cluster distances and the between-cluster distances. %

We propose a third approach when the within-cluster and between-cluster distances substantially overlap. Now, the distribution of pairwise distances %
has no well-defined local minimum. %
We propose running $k$-means clustering on the pairwise distances over a range of values of $k$. For each $k$, discard clusters assigned fewer than some fraction of the total number of datapoints, and compute the minimum pairwise distance among the surviving clusters. Call this $d_{\min,k}$.
We propose setting the thinning radius as \\
$
\hspace*{.6in}	R = (d_{\min, k^{\prime}} + d_{\min, k^{\prime}+1})/2, \quad \text{where }
    k^{\prime} = \underset{2\leq k\leq k_{\max}-1}{\operatorname*{argmax}}(d_{\min,k} - d_{\min,k+1}).
$\\
The rationale is to look for a sudden drop in the minimum intercluster distance, and use this to set typical cluster separation $R$. %

To illustrate this, we consider the following two-dimensional Gaussian mixture model:\\
\renewcommand{\arraystretch}{0.6}
\setlength{\arraycolsep}{3pt}
$
 \hspace*{.4in}   y_{1},...,y_{n} \stackrel{iid}{\sim} 0.75 N_{2}\left(\begin{pmatrix} 0\\ 0\end{pmatrix},\begin{pmatrix} 4 & 3.2\\ 3.2 & 3\end{pmatrix}\right) + 0.25 N_{2}\left(\begin{pmatrix} 3\\ 3\end{pmatrix},\begin{pmatrix} 3 & -2.1\\ -2.1 & 3\end{pmatrix}\right)
$\\
\renewcommand{\arraystretch}{1}
\setlength{\arraycolsep}{5pt}
where $N_{d}(\boldsymbol{\mu},\boldsymbol{\Sigma})$ ($d\geq 2$) denotes a $d$-dimensional Gaussian distribution with mean vector $\boldsymbol{\mu}$ and covariance matrix $\boldsymbol{\Sigma}$. We simulated a training dataset of size 600 and a test data with 300 observations were simulated independently from this.
We model this dataset as a MRMM, with prior $p_{\Theta}(\theta)$ a Gaussian with mean zero and covariance $10I_{2}$, and an inverse Wishart prior with 2 degrees of freedom and a scale matrix $I_{2}$ on the covariances.

\begin{figure}[tbh!]
\begin{minipage}{0.55\textwidth}
	\centering
	\includegraphics[width=\textwidth]{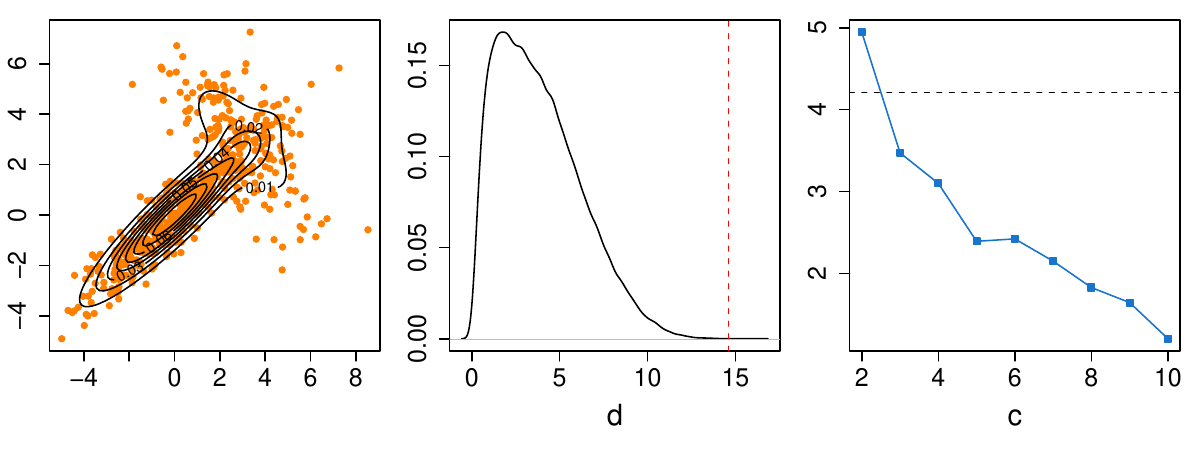}
    \includegraphics[width=\textwidth]{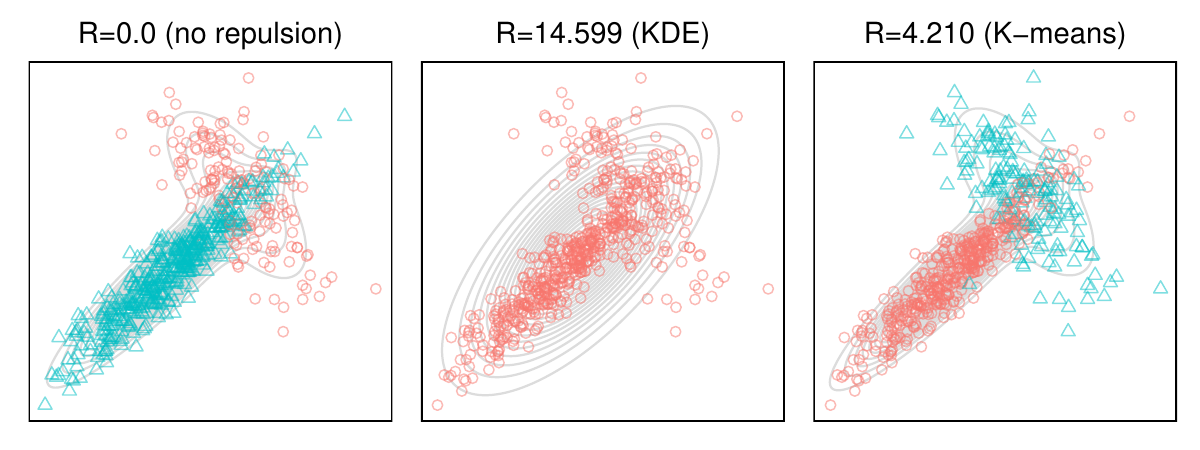}
    \end{minipage}
\begin{minipage}{0.38\textwidth}
	\caption{
	Top: (Left) Scatterplot of data with true mixture density. (Middle) Kernel density estimate of pairwise distances (Right)  $d_{\min,k}$ versus $k$.
    Bottom: Contour plot and cluster assignments of the bivariate data for hardcore MRMM.}
	\label{fig_EB_results_2d}
\end{minipage}
\end{figure}

The top-middle panel of Figure~\ref{fig_EB_results_2d} shows the kernel density estimate of pairwise distances. The distribution is unimodal, resulting in a relatively large estimated thinning radius of \(14.599 \). The top-right panel of Figure~\ref{fig_EB_results_2d} presents the minimum pairwise distance between cluster centers. We observe a pronounced drop in \( d_{\min,c} \) when increasing the number of clusters from \( c = 2 \) to \( c = 3 \), suggesting the emergence of a redundant cluster when fitting mixture models with \( c \geq 3 \). In this case, the estimated thinning radius is \( 4.210 \). %
The results for the hardcore MRMM are summarized in the bottom panel of~\Cref{fig_EB_results_2d} and Table~\ref{table_hardcore_2d}. %
\begin{table}[!h]
\centering
\footnotesize
\begin{tabular}{c|ccc|cc}
\hline
 Repulsion strength & $\EC$    & $\VarC$ & $\hC$ & $\lnpTest$ &	LPML\\
\hline\hline
\cellcolor{gray!0}{$R=0.0$ (no repulsion)} & \cellcolor{gray!0}{2.30} & \cellcolor{gray!0}{0.2642} & \cellcolor{gray!0}{2} & \cellcolor{gray!0}{-1103.53} & \cellcolor{gray!0}{-2224.42}\\\hline
$R=14.599$ (KDE-based) & 1.00 & 0.0000 & 1 & -1252.76 & -2492.42\\\hline
\cellcolor{gray!0}{$R=4.210$ (K-means-based)} & \cellcolor{gray!0}{2.02} & \cellcolor{gray!0}{0.0158} & \cellcolor{gray!0}{2} & \cellcolor{gray!0}{-1103.11} & \cellcolor{gray!0}{-2224.11}\\
\hline
\end{tabular}
\caption{Posterior summaries of hardcore MRMM on the bivariate dataset.}
\label{table_hardcore_2d}
\end{table}

The ideas above can also be used to set a hyperprior the parameters of the repulsive kernel. %
Now, we center the hyperprior $p(R)$ over the separation value identified as described above.
While this has the advantage of learning (rather than fitting) the thinning parameters, we mention that it is important to use a relatively informative hyperprior, since otherwise the model can revert to no repulsion to maximize the data fit.

\subsection{Chicago 2019 homocide data} \label{sec:exp-chicago}

We next consider a dataset of homicide recordings, collected in Chicago, Illinois in the year 2019\footnote{obtained from \url{https://data.cityofchicago.org/Public-Safety/Crimes-2019/w98m-zvie}}. 
This  consists 501 entries, which we randomly split into 416 (85\%) training data points and 85 (15\%) testing data points. 
\Cref{fig:crime-hard}(left) shows the training data, consisting of the latitude and longitude of each homicide. %
These range from $(-87.8066, -87.5293)$ to $(41.6572, 42.0208)$, and we modeled their spatial distribution with MRMM, specifically, a two-dimensional Gaussian mixture model with hardcore repulsion. %
We set $p_\Theta(\s)$ to a Gaussian density, with mean $(-87.6727, 41.8180)$ (centered in Chicago), and with variance set to $7\times10^{-3}I_2$ (to cover the entire city).
We placed an inverse-Wishart prior with 2 degrees of freedom and scale matrix $3.5\times10^{-3}I_2$ on the covariance of each Gaussian mixture component.
In settings where we wished to learn the thinning radius $R$, we placed a $\dGamma{40, 200}$ prior on $R$, corresponding to {a prior mean of $0.2$ and variance of $0.001$.}

\begin{figure}[H]
    \raisebox{0.4cm}{
	\includegraphics[width=0.17\linewidth]{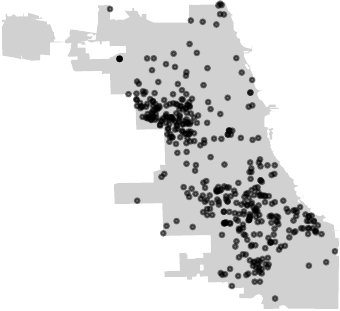}
    }
	\includegraphics[width=0.82\linewidth]{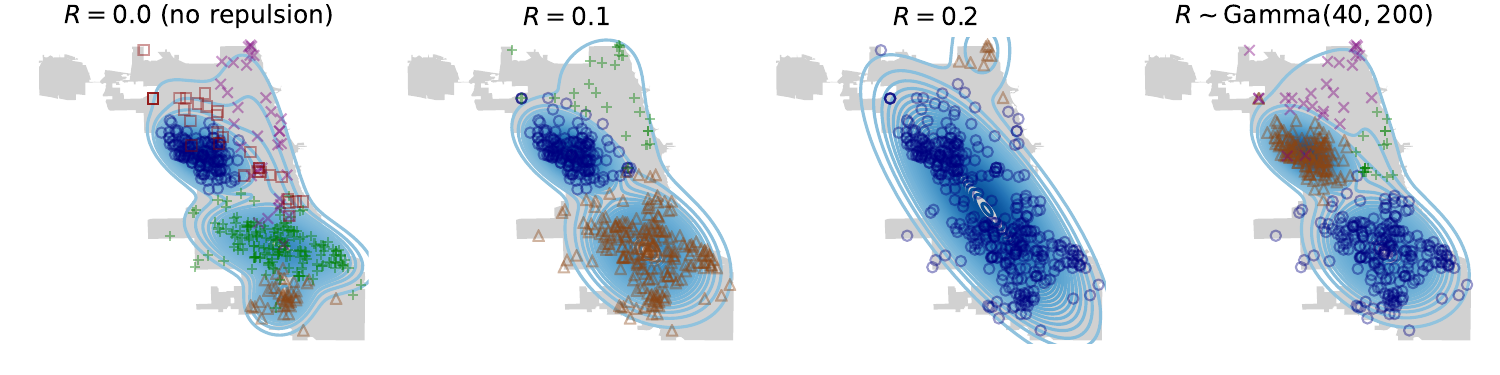}
	\caption{
		Chicago crime data, with contours/component assignments of hardcore MRMM.}
	\label{fig:crime-hard}
\end{figure}
\Cref{fig:crime-hard} and \Cref{tbl:crime-hard} show the results from the hardcore MRMM with different thinning radii. 
Across all posterior samples, there were 3 dominant components, with the remaining components accounting for a small portion of observations. 
Without any repulsion, the observations to the south of Chicago are assigned to three components.
Increasing the repulsion radius to $0.1$ simplifies these three components into a single large component,
even though the observations here deviate slightly from the Gaussian assumption, illustrating the robustness of MRMM to model misspecification. 
\Cref{tbl:crime-hard} shows that this simpler model does not come at the cost of a serious loss in predictive power.
Increasing the thinning radius to $0.2$ on the other hand causes a steep drop in predictive performance, with a majority of the data points now being assigned to a single component (with a few observations to the north-east assigned to their own component).
Inferring the thinning radius results in a posterior mean and variance $\EE\sgiven R\XX = 0.08$, $\Var\rgiven R\XX = 0.0001$,
and achieves a good trade-off between parsimony and goodness-of-fit.
Here again, south Chicago is covered by a  single component instead of multiple components as in the no-repulsion case. 

Similar results using probabilistic thinning are included in the supplementary material. 
One takeaway of this and subsequent experiments is that the more complicated probabilistic and softcore thinning mechanisms discussed in~\citet{rao2017matern} are not necessary in mixture modeling applications.
This is due to the fact that the number of mixture components is much smaller than the number of observations.
Consequently, simple hardcore thinning will suffice, and is typically preferable, since it more strongly enforces parsimony.

\begin{table*}
	\centering
    \footnotesize
	\begin{tabular}{c|ccc|cc}
		\hline
		Repulsion strength      & $\EC$    & $\VarC$ & $\hC$ & $\lnpTest$ &	LPML	 \\ \hline\hline
		$R=0.0$ (no repulsion) &    5.20 &    0.4028 &       5 &   252.54 &  1349.08 \\\hline
		$R=0.1$ &    3.51 &    0.2859 &       3 &   248.72 &  1312.30 \\\hline
		$R=0.2$ &    2.00 &    0.0000 &       2 &   232.95 &  1223.68 \\\hline
		$R\sim\dGamma{40, 200}$ &    3.68 &    0.2416 &       4 &   248.51 &  1318.73 \\\hline
	\end{tabular}
	\caption{
    Posterior summaries of hardcore MRMM on Chicago crime dataset.}
	\label{tbl:crime-hard}
\end{table*}

\subsection{Protein structure data}\label{sec:exp-protein}
The Malate dehydrogenase protein dataset, %
publicly available as \texttt{7mdh} in the protein data bank \citep{berman2002protein}, consists of 500 pairs of torsion angles, each pair 
$x=(\phi,\psi)\in [-\pi, \pi)\times[-\pi, \pi)$ forming a point on a torus.
\Cref{fig:protein-data} plots this data, with the right panel showing a planar representation known as the Ramachandran plot \citep{Ramachandranetal1963}.
While the latter shows the underlying clustering structure, it ignores the fact that the edges wrap back to each other, 
making common distributions on two-dimensional Euclidean spaces (e.g.\ mixture of normals or Betas) inappropriate.
Instead, we model this data as a mixture of uncorrelated bivariate von Mises distributions~\citep{mardia1975statistical}.
\begin{figure}[H]
  \begin{minipage}[!hp]{0.32\linewidth}
	\includegraphics[width=.99\linewidth]{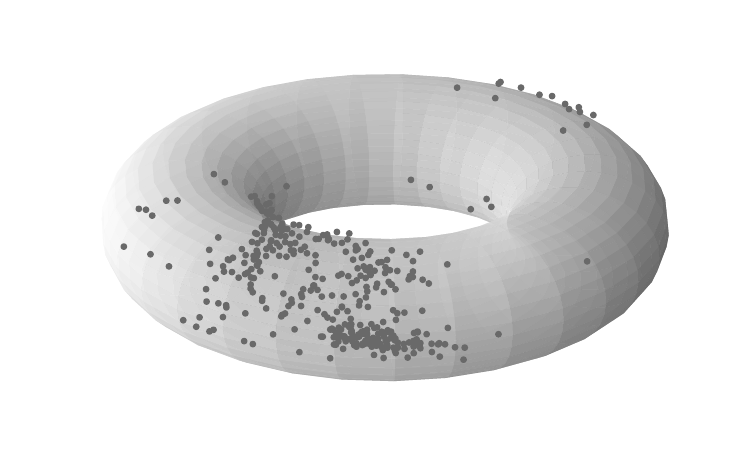}
  \end{minipage}
  \begin{minipage}[!hp]{0.3\linewidth}
	\includegraphics[width=.99\linewidth]{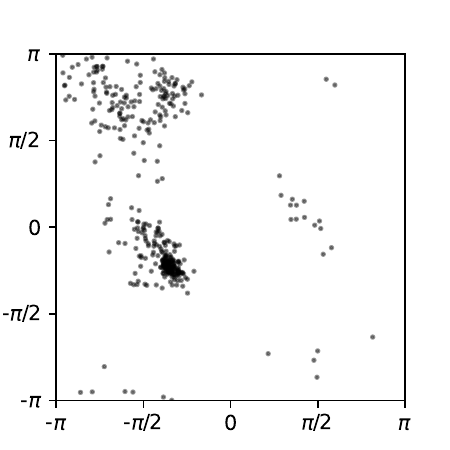}
  \end{minipage}
  \begin{minipage}[!hp]{0.35\linewidth}
	\caption{
		The Malate dehydrogenase protein data, plotted  
		\textbf{(Left)} on a torus. \textbf{(Right)} as a Ramachandran plot, where the torus is flattened to 2-d.}
	\label{fig:protein-data}
  \end{minipage}
\end{figure}

The univariate von Mises distribution %
has density
$
	\pgiven{\phi}{\mu, \kappa} = \frac1{2\pi I_0(\kappa)}\exp\cbr{\kappa \cos(\phi-\mu)},
$ for $\phi\in[-\pi, \pi)$:
here $\mu$ is the center (mean and mode), $\kappa > 0$ measures concentration around this, and $I_0(\cdot)$ is the modified Bessel function of the first kind of order $0$. 
This distribution is analogous to the univariate Gaussian distribution in the Euclidean space, though it captures the periodicity of the angular variables.
It converges to the uniform distribution on $[-\pi, \pi)$ when $\kappa\rightarrow0$. 
Writing each observation as $x=(\phi,\psi)$, we model these using a \matern repsulsive mixture model, where under each mixture component, the angles $\phi$ and $\psi$ are independent von Mises variables.
Write the parameters of each mixture component as $\s=(\mu_1, \mu_2)$ and $\kappa=(\kappa_1, \kappa_2)$, then observations from that component have density
$
	\px{x=(\phi,\psi)}{\s, \kappa} \propto 
	\exp\cbr{\kappa_1 \cos(\phi-\mu_1) + \kappa_2 \cos(\psi-\mu_2)}.
$
We set $p_\Theta(\s)$ to the bivariate uniform distribution on $[-\pi.\pi]\times[-\pi,\pi]$, %
and placed a $\dGamma{10, 1}$ %
prior on the concentration parameter $\kappa$.
To induce \matern thinning, we computed distances on the torus as 
$
	d_2((\phi, \psi), (\phi',\psi')) = \sqrt{d_1(\phi,\phi')^2 + d_1(\psi,\psi')^2},
$
where $d_1(\phi,\phi') = \min\cbr{|\phi-\phi'|, \pi-|\phi-\phi'|}$. 
This distance was used in a standard harcore or probabilistic thinning kernel.

 \Cref{fig:protein-hard} and \Cref{tbl:protein-hard} show the results with different levels of repulsion. Observe from~\Cref{fig:protein-data} that the data consists three large components of observations, with a couple of smaller components. Our model without repulsion returns about 12 components on average under the posterior distribution, with the leftmost panel of~\Cref{fig:protein-hard} showing the median clustering.
\begin{figure}[H]
	\includegraphics[width=\linewidth]{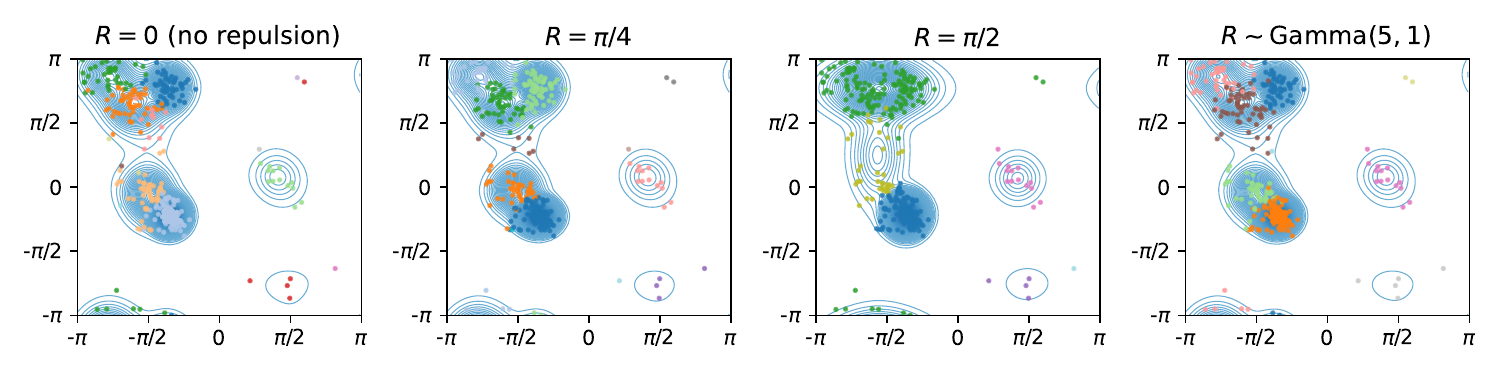}
	\caption{
		Contours  and cluster assignments of the protein data from hardcore MRMM.}
	\label{fig:protein-hard}
\end{figure}
\begin{table*}
	\centering
    \footnotesize
	\begin{tabular}{c|ccc|cc}
		\hline
		Repulsion strength      & $\EC$    & $\VarC$ & $\hC$ & $\lnpTest$ &	LPML	 \\ \hline\hline
		$R=0$ (no repulsion) &   12.29 &    4.1242 &      14 &  -177.43 & -644.23 \\\hline
        $R=\pi/4$ &   10.22 &    1.6658 &      12 &  -177.52 & -646.58 \\\hline
		$R=\pi/2$ &    5.55 &    0.3999 &       6 &  -199.78 & -703.62 \\\hline
		$R\sim\dGamma{5, 1}$ &   11.13 &    2.5746 &       9 &  -177.76 & -647.00 \\\hline
	\end{tabular}
	\caption{
    Posterior summaries of hardcore MRMM on the Malate protein dataset.}
	\label{tbl:protein-hard}
\end{table*}
As with the Euclidean setting, increasing repulsion strength results in fewer components, simpler posterior distributions (indicated by smaller posterior variance) and more interpretable results.
A strong repulsion ($R=\pi/2$) produced around 5 components, agreeing with the findings in \citet{mardia2007protein}, though resulting in a drop in model fit and predictive power. 
Placing a $\dGamma{5, 1}$ prior (mean 5, variance 5) on the thinning radius infers weaker repulsion (a posterior mean and variance for $R$ equal to $0.19\pi$ and $0.017\pi^2)$, and thus more components ($11$ on average).
These results are partly because of our choice of component likelihoods, where the two angles are independent under each component.
The component near the origin on the other hand exhibits strong correlation between the angles, and our MRMM model has to split this into two (\Cref{fig:protein-hard}, right).
We can easily extend our model %
so each component is a bivariate von Mises distribution with correlations, or use geodesic distances~\citep{mardia1975statistical, mardia2007protein}.
The former however introduces intractable normalization constants, and to avoid unnecessary complications~\citep{rao2016data, lin2017}, we have not followed this path.
We emphasize that modeling repulsion on non-Euclidean spaces using existing models is a less straightforward proposition.

\subsection{Comparison with \citet{xie2019bayesian} on the Old Faithful dataset} \label{sec:exp-faithful}

The Old Faithful  dataset~\citep{silverman1986density},
recording eruption lengths of the Old Faithful geyser in the Yellowstone National Park, was used by~\citet{xie2019bayesian} to evaluate their model, and here, we use it to compare our model with theirs.
Following \citet{xie2019bayesian}, we paired each eruption duration time with the time length of the next, resulting in 271 bivariate observations.
We split this into training and test sets of size $219$ and $52$. 

As in the setup of \citet{xie2019bayesian}, we used a Gaussian $p_\Theta(\s)$, centered at $(0, 0)$, and with covariance $10I_2$. 
For the covariance matrix of each mixture component, \citet{xie2019bayesian} assumed independence between the two dimensions and placed truncated inverse $\dGamma{1, 1}$ priors on the diagonal elements. 
We used the more natural inverse-Wishart prior with 2 degrees of freedom and scale matrix $I_2$ on the covariance matrices. 
We set the repulsive parameter of~\citet{xie2019bayesian} to its default setting of their code (also the setting in their experiments).
\Cref{tbl:faithful-hard} and \Cref{fig:faithful-hard} report posterior summaries of both models. 

This dataset consists of four clearly separated components, and for all models, the posterior mean of the number of components was around this value.
MRMM returns slightly higher estimates compared to \citet{xie2019bayesian}, but with a much smaller sample variance, suggesting a simpler, more concentrated posterior. 
So long as the thinning radius is not forced too large, MRMM also returns better fits, both in terms of predictive likelihood and LPML.
Both the model of~\citet{xie2019bayesian} and MRMM with $R=2$ merge the two top components into a large component, whereas other settings of MRMM keep them separated. 
This is also the case with a $\dGamma{4, 2}$ prior on $R$, here the thinning radius has posterior mean $\EE\sgiven R\XX = 1.40$ and variance $\Var\rgiven R{\XX} = 0.1864$, with an average of $4$ components.
\begin{figure}[H]
	\includegraphics[width=\linewidth]{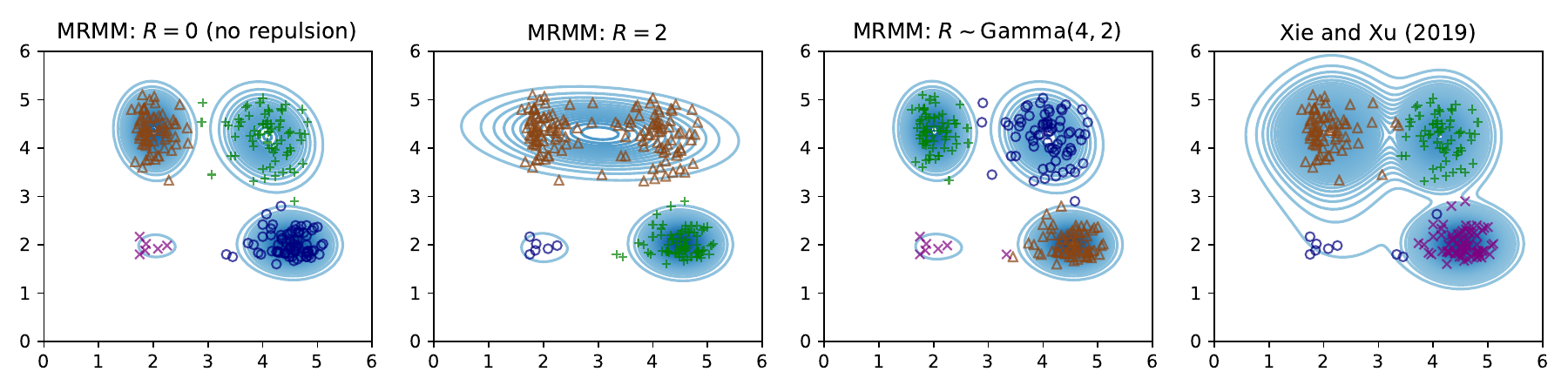}
	\caption{
		Contours and cluster assignments of Old Faithful dataset with hardcore MRMM.}
	\label{fig:faithful-hard}
\end{figure}
\begin{table*}
{\footnotesize
	\centering
	\begin{tabular}{l|ccc|cc|c|c}
		\hline
        Model & $\EC$ \hspace{-.2in} & $\VarC$  \hspace{-.2in} & $\hC$ & $\lnpTest$ \hspace{-.2in} &  LPML & Runtime (s) & ESS/s    \\ \hline\hline
		\citet{xie2019bayesian} &    3.71 &    0.212 &       4 &  -104.32 & -464.22 &  225.6 & 0.01\\\hline
		MRMM, $\ R=0$ &    4.02 &    0.018 &       4 &   -95.80 & -421.17 &  266.5 & 0.67 \\\hline
		MRMM, $\ R=2$ &    3.00 &    0.000 &       3 &  -114.84 & -489.83 &  251.1 & 5.54 \\\hline
		MRMM, $\ R\sim\dGamma{4, 2}$ &    4.01 &    0.012 &       4 &   -95.77 & -420.54 &  279.4 & 0.07 \\\hline
	\end{tabular}
	\caption{
    Posterior summaries of hardcore MRMM on the Old Faithful geyser eruption data.}
	\label{tbl:faithful-hard}
    }
\end{table*}

As \Cref{tbl:faithful-hard} shows, while both models required roughly the same time per iteration (though  
implemented in Matlab and Python),
mixing in their case was poorer, and we had to run their algorithm for twice the number of iterations as ours to get stable results. 
This can be seen in the ESS/s numbers, where our sampler shows (often much) larger  values. 
Similar results for probabilistic MRMM are in the supplementary material.

\subsection{Comparison with \citet{bianchini2018determinantal} on the Galaxy dataset} \label{sec:exp-galaxy}
\begin{figure}[H]
	\includegraphics[width=\linewidth]{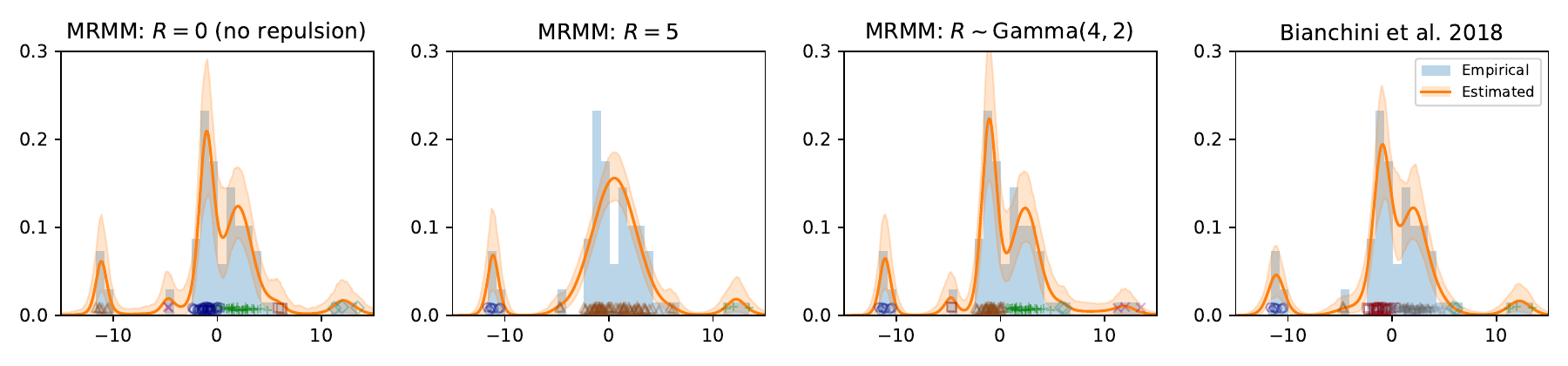}
	\caption{
		Contour plot and cluster assignments of the Galaxy data for hardcore MRMM.}
	\label{fig:galaxy-hard}
\end{figure}
The Galaxy dataset~\citep{roeder1990density} available from the \texttt{DPpackage} in \texttt{R}, contains the velocities of 82 different galaxies. %
\citet{bianchini2018determinantal} evaluated their model on this well-known dataset, 
using LPML as their goodness-of-fit criteria. 
We use the same here. %
Following the same steps as \citet{bianchini2018determinantal}, we centered the data and rescaled it by a factor of $10^{-3}$, %
set $p_\Theta(\s)$ to a  mean 0 and standard deviation 10 Gaussian, and placed an inverse-$\dGamma{3, 3}$ prior on the variance of each mixture component.

\begin{table*}
{\footnotesize
	\centering
	\begin{tabular}{l|ccc|c|c|c}
		\hline
		Model & $\EC$  & $\VarC$  & $\hC$ &  LPML & Runtime (s) & ESS/s   \\ \hline\hline
		\citet{bianchini2018determinantal} &    6.00 &    1.218 &       7 & -207.94 &   600.4 & 0.02 \\\hline
		MRMM, $\ R=0$ (no repulsion) &    7.69 &    4.082 &       6 & -210.13 &   772.9 & 0.83 \\\hline
		MRMM, $\ R=5$ &    3.37 &    0.305 &       3 & -212.05 &   448.2 & 4.50\\\hline
		MRMM, $\ R\sim\dGamma{4, 2}$ &    5.51 &    0.934 &       6 & -208.83 &  501.2 & 0.03\\\hline
	\end{tabular}
	\caption{
		Posterior summaries for the Galaxy dataset inferred with hardcore MRMM. 
      }
	\label{tbl:galaxy-hard}
    }
\end{table*}

The left-most panel in~\Cref{fig:galaxy-hard} dispays the mean posterior density for MRMM without any repulsion. 
The posterior mean number of components is around 8, with a relatively large variance of 4. The two rightmost panels show the corresponding densities for MRMM (with the thinning radius learnt), and the model of~\citet{bianchini2018determinantal}. 
Both have about 6 components, though the predictive performance is not significantly different from the model without repulsion.
By forcing the thinning radius to $5$, the components around the origin merge into a single component, 
with noticeable, but not large drop in performance.
With a Gamma %
prior on $R$, we get a posterior mean $\EE\sgiven R\XX = 1.54$ and variance $\Var\rgiven R{\XX} = 0.5305$, with the posterior mean of the number of components about $5.5$.
We report comparable CPU run times and ESS/s of both methods in \Cref{tbl:galaxy-hard}.

\vspace{-.1in}
\section{Discussion} \label{sec:conc}
In this paper, we described a novel approach to repulsive mixture modeling through the \matern type-III repulsive point process. The advantages of our approach include its mechanistic nature, which allows easy extension to different kinds of repulsion, as well as the simplicity and efficiency of the  associated MCMC sampling algorithm. %
While we only considered repulsion between component locations, it is also of interest to consider repulsion between variances or even cluster weights. %
From a theoretical viewpoint, better understanding the effect of the \matern kernel parameters on the %
repulsive behavior of our model will provide practitioners with an additional tool for model selection. 
It is also of interest to investigate asymptotic properties such as posterior consistency and convergence rates of this class of repulsive mixture models. %
Another extension to %
high-dimensional clustering applications (e.g. through random projections), into more general models such as latent feature models or time series models such as self-avoiding Markov models.
Finally, it is of interest to apply these models to new applications and problems.

\bibliographystyle{abbrvnat}
\bibliography{Matern.bib}

\newpage

\begin{center}

{\large\bf SUPPLEMENTARY MATERIAL}

\end{center}

\setcounter{page}{1}
\setcounter{section}{0}
\setcounter{figure}{0}
\setcounter{table}{0}
\renewcommand{\thesection}{\Alph{section}}
\renewcommand{\thefigure}{S\arabic{figure}}
\renewcommand{\thetable}{S\arabic{table}}

\begin{center}
	{\Large\bf 
		Supplementary material for \\
		``Bayesian Repulsive Mixture Modeling with 
		\\ \matern Point Processes"}
\end{center}

These supplementary materials include the detailed proofs, algorithms and additional experiment results.

\section{Proofs}

\begin{theorem*}[\ref{prop:X-G-Gt}]
       Write $\scP_\lambda$ for the law of a rate-$\lambda(\cdot)$ Poisson process on $\Space\times \mathcal{M}$. Then the measure of the tuple $\XX$, $\G$, $\Gt$ has density with respect to ${d}x^n\times\scP_\lambda$ given by
	\begin{align*}
		\pgiven{\XX, \G, \Gt}{\Rate, \thin} & =
		\left(\frac{\ind{}(|\G \cup \Gt| > 0)}{1-e^{\int_{\Space}- \lambda(\s, \w, \t)\,\dif\s \dif\w \dif\t }}\right) \\
		&  \left( \prod_{\g\in\G}\sbr{1-\shadow\g\G}\prod_{\gt\in\Gt}\shadow\gt\G  \right)
		\left( \prod_{i=1}^{n}\sum_{(\s, \w, \t)\in\G}{\frac\w{\sum\G_\Weight}\px{x_i}{\s} } \right).
	\end{align*}	
\end{theorem*}
\begin{proof}
	First note that the set $\F = \G \cup \Gt$ follows a Poisson process with rate $\lambda(\s,\w,\t)$, conditioned to have at least 1 event. 
	The probability that such a Poisson process produces $1$ or more events is $1-\exp(-\int \lambda(\theta,\w, \t) d \theta  d\w d\t)$.
	It follows that conditioning on this event, $\F$ has density with respect to $\scP_\lambda$ given by the ratio in the first parentheses.
	Each element $f$ of $F$ is assigned to either $\G$ or $\Gt$, with probability 
	${1-\shadow\f\G}$ or $\shadow\f\G$ respectively. 
	This gives the terms in the second parentheses.
	Finally, the $i$th observation is assigned to cluster $(\s,\w,\t) \in G$ with probability $\w/\G_\Weight$, with its value having density $\px{x_i}{\s}$ with respect to $dx$.
	Marginalizing over cluster assignments, and considering all $n$ observations, we get the final terms.
	The result then follows easily from Lemma \ref{lem:poiss_density}.
\end{proof}

To prove~\Cref{prop:Gt}, we start with the following useful (and not new) result. Below, we give a less combinatorial and slightly more general proof than what~\citet{rao2017matern} used implicitly in their work: 
\begin{lemma} \label{lem:poiss_density}
	Consider two Poisson processes on some space $\cY$, with intensities $\lambda(y)$ and $\mu(y)$. Then the former has density with respect to the latter given by %
	\begin{align}
		\frac{d \scP_\lambda}{d \scP_\mu}(M) := p_\mu(M|\lambda) = e^{\int_{\cY} \mu(y) - \lambda(y)\dif y }\prod_{m\in M}\frac{\lambda(m)}{\mu(m)}%
	\end{align}
\end{lemma}
\begin{proof}
	Consider a function $h:\cY \rightarrow \Re$.
	For a point process $M$ on $\cY$, we overload notation, and define the linear functional $h(M) = \sum_{m \in M} h(m)$.
	Write $\EE_\scM[h(M)]$ for the expectation of $h(M)$ when $M$ is distributed as a point process with measure $\scM$.
	Recall that $\scP_\lambda$ corresponds to a rate-$\lambda(\cdot)$ Poisson process on $\cY$, and $\scP_\mu$, to a rate-$\mu(\cdot)$ Poisson process. %
	We first note that from Campbell's theorem~\citep{kingman1992poisson}, for a rate-$\mu(\cdot)$ Poisson process, we have %
	\begin{align}
		\label{eq:laplace}
		\EE_{\scP_\mu}[\exp(h(M))] 
		= \EE_{\scP_\mu}\sbr{\exp\bigg(\sum_{m \in M} h(m)\bigg)} 
		= \exp \left( \int (e^{h(y)}-1) \mu(y)\dif y  \right).
	\end{align}
	Now write $\scM^\lambda_\mu$ for the probability measure of a point process with density $p_\mu(M|\lambda)$ with respect to a rate-$\mu(\cdot)$ Poisson process. Then %
	\begin{align}
		\EE_{\scM^\lambda_\mu}[\exp(h(M))] 
		& = \EE_{\scP_\mu}\left[p_\mu(M|\lambda) \exp(h(M))\right]  \nonumber \\
		& = \EE_{\scP_\mu}\left[e^{\int_{\mathcal{Y}} \left(\mu(y) - \lambda(y)\right)\dif y }\left( \prod_{m \in M}\frac{\lambda(m)}{\mu(m)} \right) \exp(h(M))\right]  \nonumber \\
		& = e^{\int_{\mathcal{Y}} \left(\mu(y) - \lambda(y)\right)\dif y } \ \EE_{\scP_\mu}\left[\exp\sum_{m \in M} \left(h(m) + \log {\lambda(m)} - \log{\mu(m)} \right)\right] \nonumber \\
		& =  \exp \left( \int_{\mathcal{Y}} (e^{h(y)}-1) \lambda(y)\dif y \right) \quad \text{\ (from~\cref{eq:laplace})} \nonumber \\
		& =  \EE_{\scP_\lambda}[\exp(h(M))].
	\end{align}
	This confirms that $\scM^\lambda_\mu$ equals $\scP_\lambda$ a.e., proving our result.
\end{proof}

\begin{prop*}[\ref{prop:Gt}]
	Given all other variables, %
	the conditional distribution of the thinned events $\Gt$ is a Poisson process with intensity
	${\rate\cdot \shadow\cdot{G}}$.
\end{prop*}
\begin{proof}
	With respect to a rate-$\lambda(\cdot)$ Poisson process,
	\begin{align}
		p(\Gt|-) & \propto  \pgiven{\G, \Gt, \XX}{\Rate, \thin} \nonumber \\
		& = \left(\frac{\ind{}(|\G \cup \Gt| > 0)}{1-e^{\int_{\Space}- \lambda(\s, \w, \t)\dif\s \dif\w \dif\t }}\right) \prod_{\g\in\G}\sbr{1-\shadow\g\G}\prod_{\gt\in\Gt}\shadow\gt\G \nonumber\\
		&  \propto \prod_{\gt\in\Gt} \shadow\gt\G. \nonumber
	\end{align} 
	In the last equation, we dropped all terms that do not depend on $\Gt$, and used the fact that since $|G|>0$, $\ind{}(|\G \cup \Gt| > 0)$.
	The result now follows from Lemma \ref{lem:poiss_density}.
\end{proof}

\section{Additional Figures}

\begin{figure}[H]
	\includegraphics[width=\linewidth]{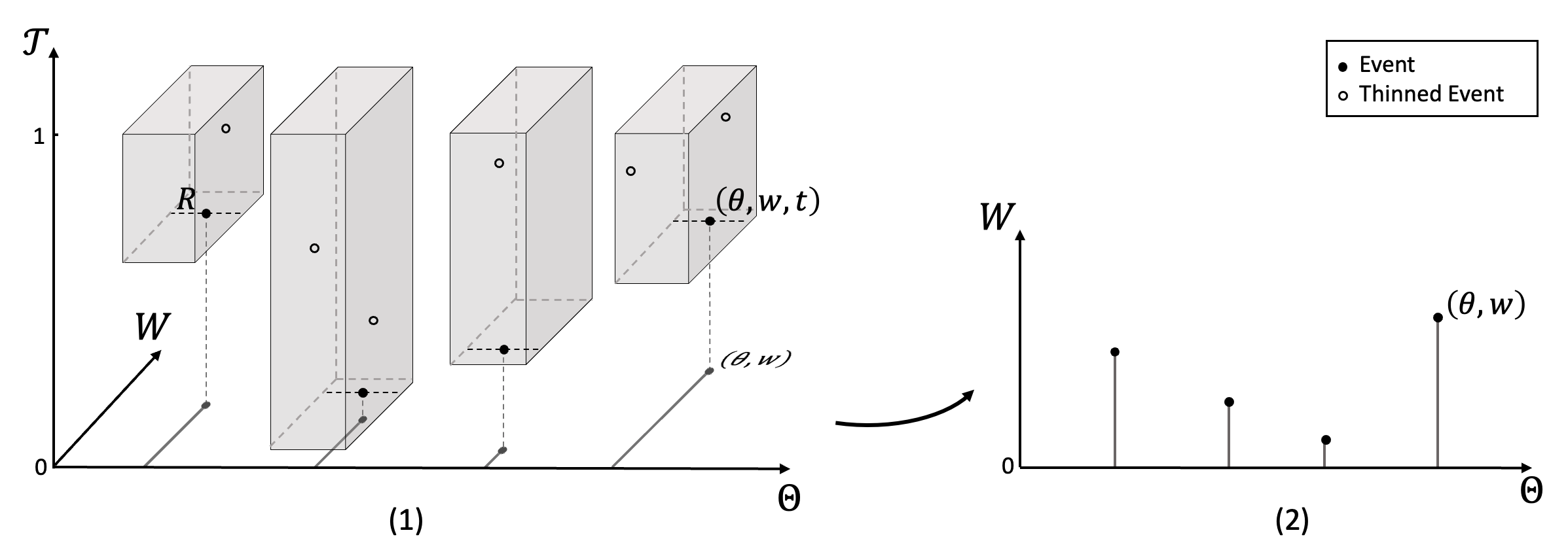}
	\caption{Illustration of the \matern prior for mixture models. 
		\textbf{(1)} Primary Poisson events $\F=\cbr{(\s_1, \w_1, \t_1), \dotsc, (\s_\CF, \w_\CF, \t_\CF)}$ thinned by a hardcore thinning kernel with thinning radius $R$. 
		The surviving events are projected to the parameter space of the mixture model $\Location\times\Weight$. 
		\textbf{(2)} The resulting mixture model, consisting of a collection of mixture component parameters $\s\in\Location$ and their corresponding unnormalized mixture weights $\w\in\Weight$. }
	\label{fig:model}
\end{figure}

\newpage
\section{Algorithms}

\begin{algorithm}
	\caption{Details of the function $\maternThin{\F, \thin}$} \label{alg:matern_thin}
	
	\DontPrintSemicolon
	\SetKwProg{Fn}{Function}{:}{}
	\SetKwInOut{Input}{Input}
	\SetKwInOut{Output}{Output}
	
	\nonl\Fn{$\maternThin{\F, \thin}$}{  %
		\Input{Extended primary Poisson process $\F$ and thinning kernel $\Kernel_\thin$}
		\Output{Extended \matern events $\G$ and thinned events $\Gt$}
		\BlankLine		
		
		Write $\overrightarrow{\F} = \rbr{\f_1,\dotsc,\f_{|\F|}}$ for $\F$ sorted in ascending order of birth times (so that $\proj{\Time}(\f_j) < \proj{\Time}(\f_{j'})$ if $j < j'$). %
		
		\For{$j \gets 1$ \KwTo $\CF$}{
			Set $(\s,\t) \gets \rbr{\proj{\Location}(\f_j), \proj{\Time}(\f_j)}$ \;
			Draw $u\sim$ Unif$[0, 1]$\; 
			\eIf(\tcp*[f]{Assign $\f_j$ to $\G$ w.p.\,$\shadow{(\s, \t)}\G$}){$u < \shadow{(\s, \t)}\G$}{
				$\G\gets\G\cup \f_j$
			}{
				$\Gt\gets\Gt\cup \f_j$
			}
		}
	}
	\Return{$\G$, $\Gt$}
\end{algorithm}

\begin{algorithm}
	\caption{The relabeling step to update \matern events $\G$} \label{alg:relabel}
	
	\DontPrintSemicolon
	\SetKwProg{Fn}{Function}{:}{}
	\SetKwFunction{mt}{{\matern}Thin}
	\SetKwFunction{rl}{Relabel}
	\SetKwFunction{sf}{RandomShuffle}
	\SetKwInOut{Input}{Input}
	\SetKwInOut{Output}{Output}
	\SetKw{Next}{next}
	
	\nonl\Fn{\rl{$\Rate$, $\au$, $\G$, $\Gt$, $\XX$}}{
		\Input{Primary Poisson intensity $\Rate$, augmentation factor $\au$, current state of the surviving events $\G$ and the thinned events $\Gt$, the data $\XX$. }
		\Output{Updated \matern events $\G$ and thinned events $\Gt$.}
		\BlankLine
		
		Sample augmented $\Ft\sim\poisP{\au\Rate(\cdot)}$\;
		Impute non-locational parameters of $\Gt$ from the prior (if presents in the model)\;
		Obtain shuffled indices $J$ = \sf{$\{1, \dots, \card{\G\cup\Gt\cup\Ft}\}$}\;
		Compute likelihood related objects: $n\times\card{J}$ matrix $L=\rbr{\w_j\px{x_i}{\s_j}: i, j}$ and $n$-dim vector $\ll=\big(\sum_{\g\in\G}l_1^\g, \dots, \sum_{\g\in\G}l_n^\g\big)$\;
		Compute the normalizing constant $\Sw=\sum\G_\Weight$\;
		\ForEach{$j$ in $J$}{
			\If{event $j$ in $\G$}{
				\eIf(\tcp*[f]{$G$ contains only event $j$}){$\CG = 1$}{
					\Next
				}{
					$\Sw \gets\Sw - \w_j$\;
					$\ll\gets\ll-L_{\cdot j}$
				}
			}
			Remove event $j$ from its original event set\;
			Assign event $j$ to $\G$, $\Gt$ or $\Ft$ with probability $P(e \in \G| -)$, $P(e \in \Gt| -)$ and $P(e \in \Ft| -)$ in \cref{eqn:relabel-posts}, respectively,\;
		}
		\Return{$\G$, $\Gt$}
	}
\end{algorithm}

\begin{figure}
	\includegraphics[width=\linewidth]{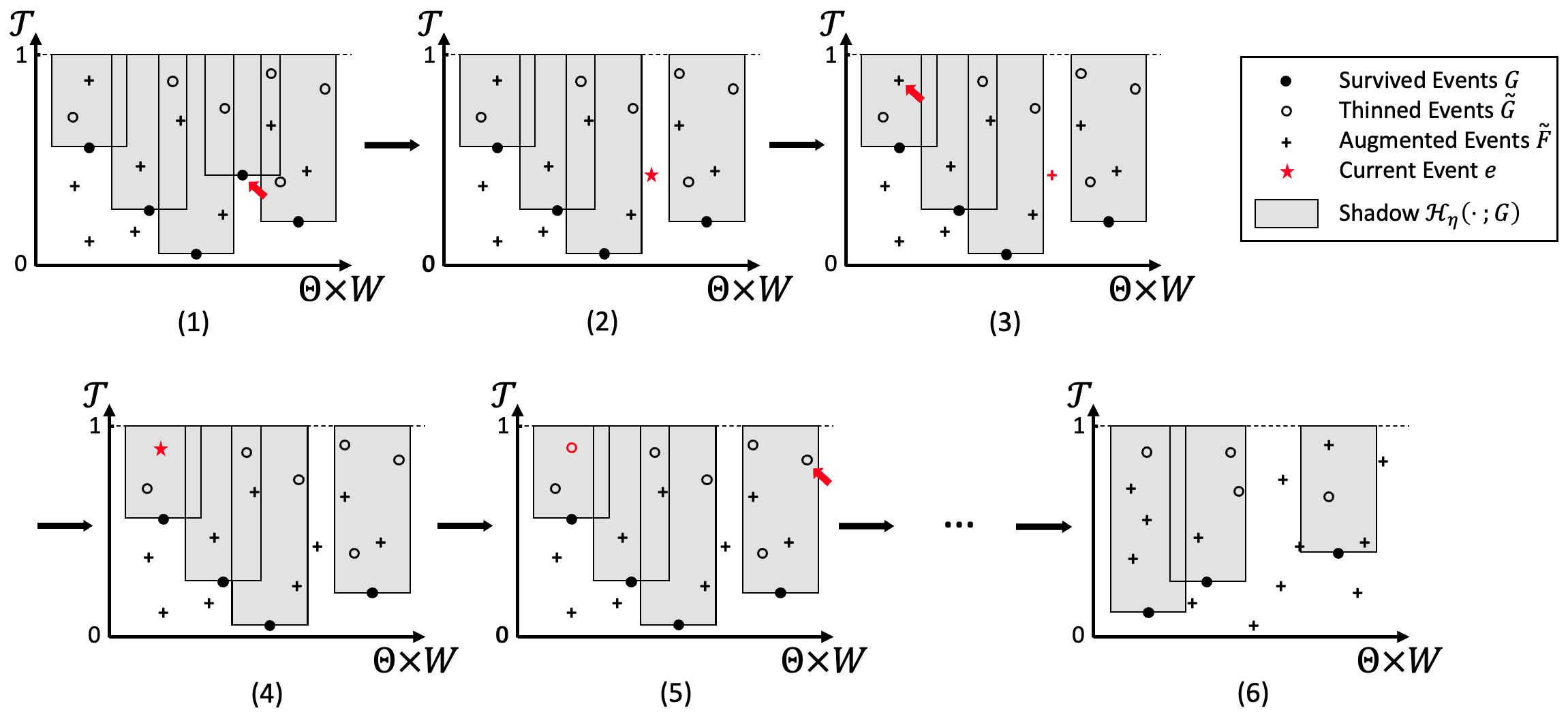}
	\caption{
		Illustration of the relabeling step.  
		\textbf{(1)} Before relabeling, the state of the surviving events $\G$, thinned events $\Gt$ and auxiliary events $\Ft$, and the shadow cast by $\G$, $\shadow\cdot\G$. 
		\textbf{(2-3)} The first event (After random shuffling of all events in $\G\cup\Gt\cup\Ft$) is relabeled as ``auxiliary". 
		The event is first removed from its original set $\G$ (and the shadow is affected accordingly) in \textbf{(2)}. 
		Then, in \textbf{(3)}, it is relabeled as ``auxiliary" according to the posterior conditional probabilities in \cref{eqn:relabel-posts}. Notice that with the hardcore thinning kernel, it is impossible for the event to be relabeled to ``thinned", as it is not under the shadow of a previously surviving event. 
		\textbf{(4-5)} The second event is relabeled as ``thinned". 
		Similarly, the event is removed from the collection of augmented events $\F$ in \textbf{(4)} and then relabeled as ``thinned" in \textbf{(5)}. Notice that it is under the shadow of a surviving event, and hence, with the hardcore thinning kernel, it can only be labeled as ``thinned" or ``auxiliary".
		\textbf{(6)} The final state for $\G$, $\Gt$, $\Ft$, after all events are relabeled. 
	}
	\label{fig:relabel}
\end{figure}

\begin{algorithm}
	\caption{Bayesian inference of MRMM} \label{alg:inference}
	
	\DontPrintSemicolon
	\SetKwFunction{rl}{Relabel}
	\SetKwInOut{Input}{Input}
	\SetKwInOut{Output}{Output}
	
	\Input{Data $\XX=\cbr{x_1, \dots, x_n}$, number of MCMC iterations $M$, model of cluster components $\px\cdot\s$, augmentation factor $\au$, prior on cluster locations $p_\s$, shape parameter of the Gamma prior on weights $\aw$, shape and rate parameter of the Gamma prior on mean intensity $(a, b)$, and prior on thinning kernel parameter $\pthin$.}
	\Output{Posterior samples of mean intensity $\mRate$, thinning parameter $\thin$, \matern events $\G_\Location$, $\G_\Time$, $\G_\Weight$, thinned events $\Gt$, and cluster assignments $\zz$. }
	\BlankLine
	
	Initialize $\mRate\sim\dGamma{a, b}$, $\thin\sim\pthin$\;
	Initialize $\G, \Gt\sim\maternP{\Rate, \Kernel_\thin}$\;
	Initialize $\zz$ from $\given\zz{\XX, \G}$: $\z_i\sim\multinomial{\w_j\cdot\px{X_i}{\s_j},~j=1,\dots,\CG}$\;
	\For{$m\gets1$ \KwTo $M$}{
		Update $\mRate$ according to $\frac1{1-e^{-\mRate}}\dGamma{a+\CF,b+1}$ using Metropolis-Hastings\;%
		Update $\thin$ according to $p({\thin}\vbar{\G, \Gt})$\;
		Update $\Gt$: (Poisson thinning) simulate from $\poisP{\Rate}$ and discard event $\gt$ with probability $1-\shadow{\gt}\G$\;
		Update $\tt$ one at a time according to \cref{eqn:tt}\;
		Update $\ww\gets \Sw\cdot\nww$ where $\Sw\sim\dGamma{\CG\aw, 1}$ and $\nww\sim\dirichlet\rbr{\aw+n_1, \dots, \aw + n_\CG}$ ($n_j=\sum_{i=1}^n \ind{}(\z_i=j))$\;
		Update $\ss$ one at a time according to \cref{eqn:ss} using Metropolis-Hastings\; %
		$\G, \Gt\gets$\rl{$\Rate$, $\au$, $\G$, $\Gt$, $\XX$}\;
		Update $\zz$ one at a time: $\z_i\sim\multinomial{\w_j\cdot\px{X_i}{\s_j}, j=1,\dotsc,\CG}$\;
	}
	\Return{Posterior MCMC samples of $\mRate$, $\thin$, $\G$, $\Gt$ and $\zz$}
\end{algorithm}

\newpage
\section{Additional Experimental Results}

\begin{table*}
	\centering
	\begin{tabular}{|l|ll|l|}
		\hline
		Thinning Kernel     & \multicolumn{2}{l|}{Thinning Parameter}   & Expression                                                \\\hline\hline
		Hardcore            & $\thin=R$       & Radius $R>0$            & $\kernelR\s{\s'}=\ind{\|\s-\s'\|<R}$                      \\\hline      
		Probabilistic       & $\thin=(R, p)$  & Radius $R>0$            & $\kernelRp\s{\s'}=p\ind{\|\s-\s'\|<R}$                    \\
		&                 & Probability $p\in[0,1]$ &                                                           \\\hline
		Squared-exponential & $\thin=l$       & Lengthscale $l>0$       & $\kernelSp\s{\s'}l=\exp\cbr{-\frac{\|\s-\s_j\|^2}{2 l}} $ \\\hline
	\end{tabular}
	\caption{Thinning kernels used in experiments}
	\label{tbl:kernels}
\end{table*}

\subsection{Effect of augmentation factor $\au$ on MCMC efficiency}
\label{sec:exp-au}

\begin{figure}[H]
	\includegraphics[width=\linewidth]{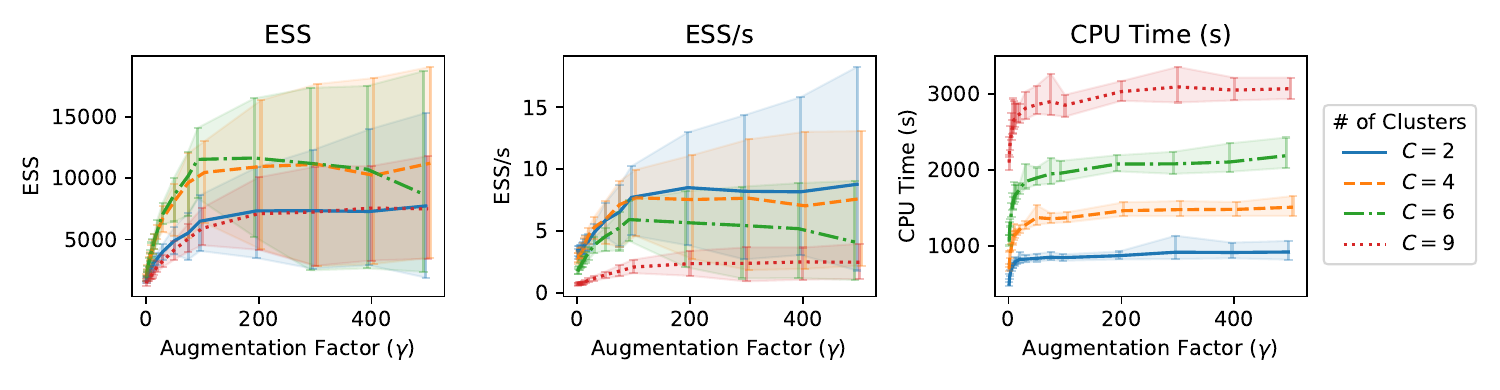}
	\caption{
		The impact of augmentation factor on \textbf{(left)} MCMC mixing (ESS out of 20,000 iterations), \textbf{(middle)} MCMC mixing rate (ESS/s) and \textbf{(right)} computational cost (CPU time). A tiny perturbation is added to $\au$'s to ensure visibility.}
	\label{fig:syn-au-2}
\end{figure}
We focus here on MRMM with hardcore thinning, the most challenging setting for MCMC mixing.
We applied MRMM to synthetic data generated from two-dimensional Gaussian mixture models, with minimum component separation of $4.0$ and with varying number of components (see the supplementary material for more details).
For each model, we simulated 50 training datasets, each consisting of 20 observations per component.  
The number of components $\C$ thus quantifies both model complexity and dataset size. 
We modeled each dataset as a hardcore MRMM with the thinning radius fixed to 2.
The covariance of each component was set to the $2\times2$ identity matrix $I_2$, and the normalized intensity $p_\Theta(\s)$  was set to $N(\mathbf{0}, 10I_2)$.
For each dataset, we ran our MCMC sampler for 20,000 iterations, with $\au$ ranging from $1$ to $500$. %

\Cref{fig:syn-au-2} plots the raw ESS (left), ESS/s (center) and  CPU run-time (right) against the augmentation factor $\gamma$, with each curve representing a different generative model. 
The right panel shows that, as expected, increasing $\gamma$ results in an increase in CPU time, as the number of events in the augmentation Poisson process increases.
At the same time, the leftmost panel shows that this added computational cost comes with the benefit of faster mixing, as more augmented Poisson events more easily allows events to be switched into and out of the \matern events $G$.
For small values for $\gamma$, this improvement is significant, before plateauing out as $\gamma$ crosses $50$.
The middle panel shows that this improvement easily compensates for the added computational burden. 
We see similar results for other thinning kernels, but do not include them.
In practice, based on these results, we recommend setting $\gamma$ somewhere in the range of $5$ to $10$. In the rest of our experiments, we fix it to $5$.

\subsection{Synthetic experiments} 

In this section, we evaluate MRMM and the associated MCMC sampling algorithm on a number of synthetic tasks.
\Cref{sec:exp-au2} provides more details and additional results on the study of augmentation factor $\au$ in \Cref{sec:exp-au}, while \Cref{sec:exp-syn} compares different thinning kernels and thinning strengths on the same synthetic datasets. \Cref{sec:eb2} provides additional experimental results on the choice of thinning parameters.

\subsubsection{Additional results for \Cref{sec:exp-au}} \label{sec:exp-au2}
The models to generate the datasets are illustrated in \Cref{fig:syn-au-data}. \Cref{fig:syn-au-mixing} visualizes assessments for the mixing of one run. 

\begin{figure} 
	\includegraphics[width=\linewidth]{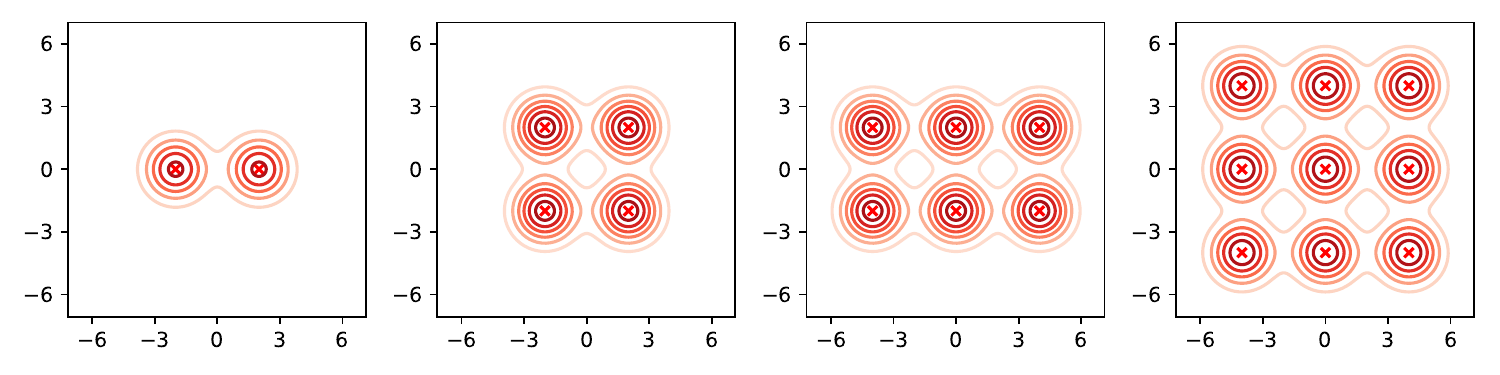}
	\caption{
		Mixtures of equally weighted Gaussian distributions for the study of augmentation factor $\au$ in \Cref{sec:exp-au}. 
		From left to right, number of clusters $\C=2, 4, 6, 9$, respectively. 
		Each cluster is a standard bivariate Gaussian with covariance being the $2\times2$ identity matrix $I_2$. 
		The minimum distances between cluster centers is 4.		
	}
	\label{fig:syn-au-data}
\end{figure}

\begin{figure}
	\includegraphics[width=\linewidth]{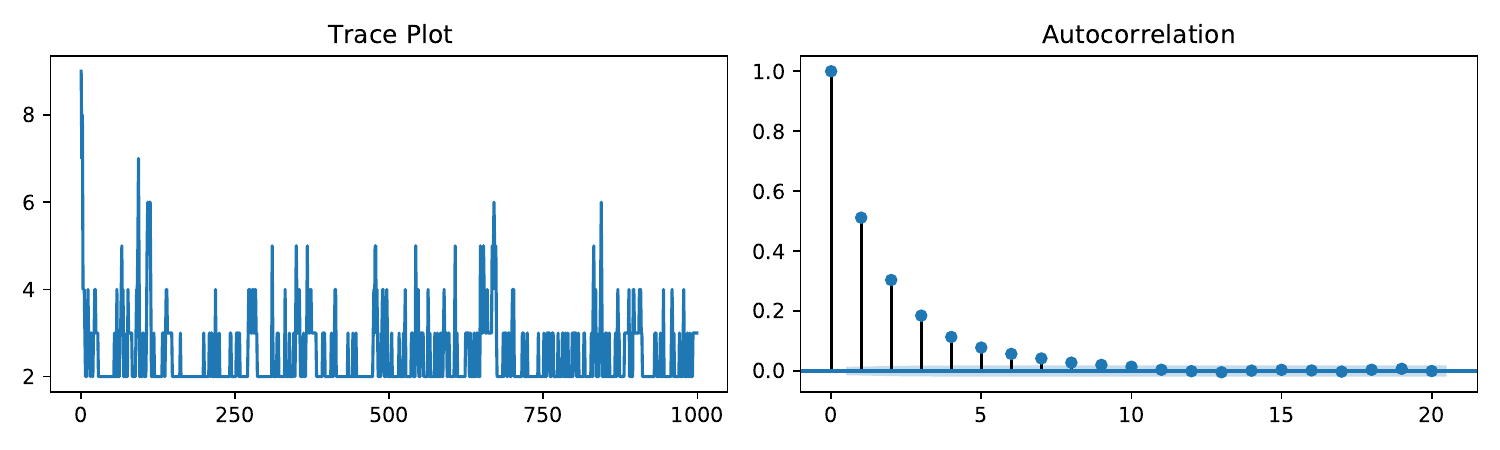}
	\caption{
		Visualization for assessing mixing of posterior number of clusters $\CG$ in one run with augmentation factor $\au=5$ on the dataset with two clusters. 
		In this run, ESS = 5624; ESS/s = 6.97; CPU Time (s) = 806.80.
		\textbf{(Left)} The trace plot of the first 1,000 updates of $\CG$. 
		\textbf{(Right)} The autocorrelation function of posterior samples of $\CG$.}
	\label{fig:syn-au-mixing}
\end{figure}

\subsubsection{Study of thinning kernels and thinning strengths} \label{sec:exp-syn}
Having established that our MCMC sampler mixes well, we now proceed to study the effect of different thinning kernels and thinning strengths on MRMM inferences. 
\Cref{tbl:kernels} lists all thinning kernels are corresponding parameters used in this study, specifically, for the probabilistic thinning kernel, the thinning probability $p=0.95$.

We consider a series of two-dimensional Gaussian mixture models shown in \Cref{fig:syn-data-all}. 
Each model consists of four equally weighted, unit-variance Gaussian components, located at $(-d/2, 3d/2)$, $(d/2, d)$, $(d, -d)$, $(-3d/2, -3d/2)$, where $d=1,2,3,4$ quantifies the separation level.
A training dataset of size 200 and a test data with 100 observations were simulated independently for each model. 

For MRMM, we set the prior $p_\Theta(\s)$ to a Gaussian with mean zero and covariance $10 I_2$. 
We placed an inverse-Wishart prior with 2 degrees of freedom and a scale matrix $I_2$ on the covariances. 
When learning the thinning strength (thinning radius $R$ for both hardcore and probabilistic MRMM, or the lengthscale $l$ for the squared-exponential MRMM), we placed a $\dGamma{4, 2}$ prior with mean 2.0 and variance 1.0.
All results were obtained from 2,000 iterations of MRMM after discarding the first 1,000 samples as burn-in.

\Cref{fig:syn-hard-all}, \ref{fig:syn-soft-all} and \ref{fig:syn-rbf-all} are the inferred posterior contours and the `median' clustering results obtained with the three kernels.
Heatmaps in \Cref{fig:syn-sum-EC,fig:syn-sum-VarC,fig:syn-sum-hC,fig:syn-sum-lnpTest-diff,fig:syn-sum-LPML} compare the parsimony and the goodness-of-fit of different thinning kernels with a variety of thinning strengths. 
As expected, increasing repulsion strength results in greater parsimony, with both the posterior mean and variance of the number of clusters dropping. 
Interestingly, moderate values of repulsion do not significantly harm the model fit.  
However, a strong repulsion strength does result in a drop in predictive power, especially for the hardcore MRMM.

\begin{figure}
	\centering
	\includegraphics[width=\linewidth]{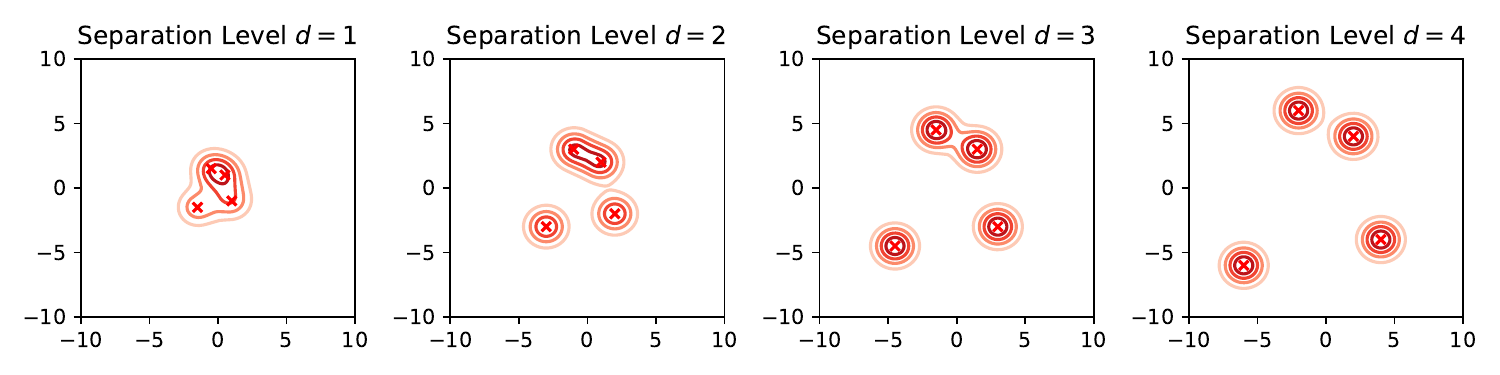}
	\caption{
		The ground truth model $M_0$ with different separation levels. 
	}
	\label{fig:syn-data-all}
\end{figure}

\begin{figure}
	\centering
	\includegraphics[width=\linewidth]{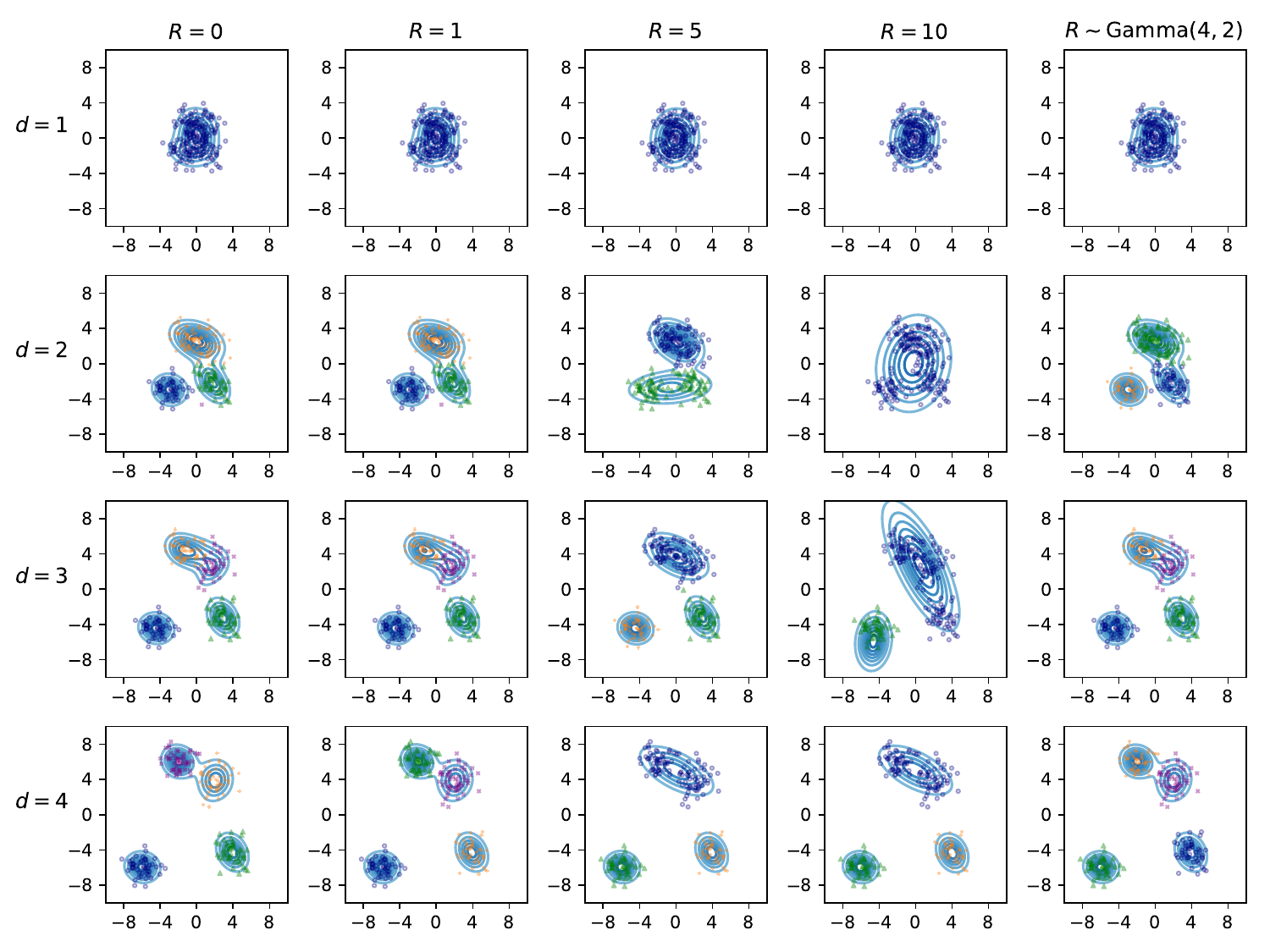}
	\caption{
		Hardcore MRMM
	}
	\label{fig:syn-hard-all}
\end{figure}

\begin{figure}
	\centering
	\includegraphics[width=\linewidth]{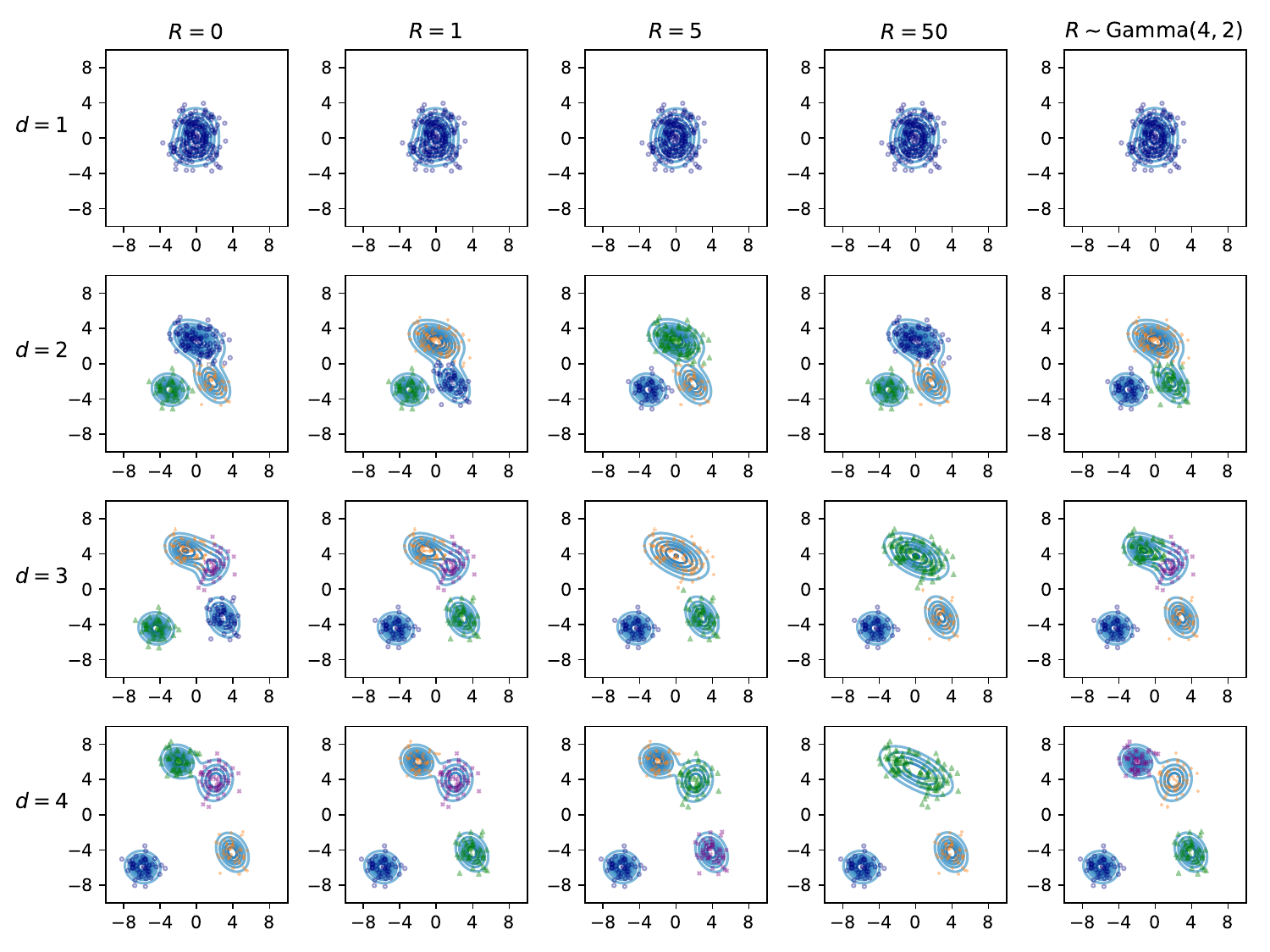}
	\caption{
		Probabilistic MRMM
	}
	\label{fig:syn-soft-all}
\end{figure}

\begin{figure}
	\centering
	\includegraphics[width=\linewidth]{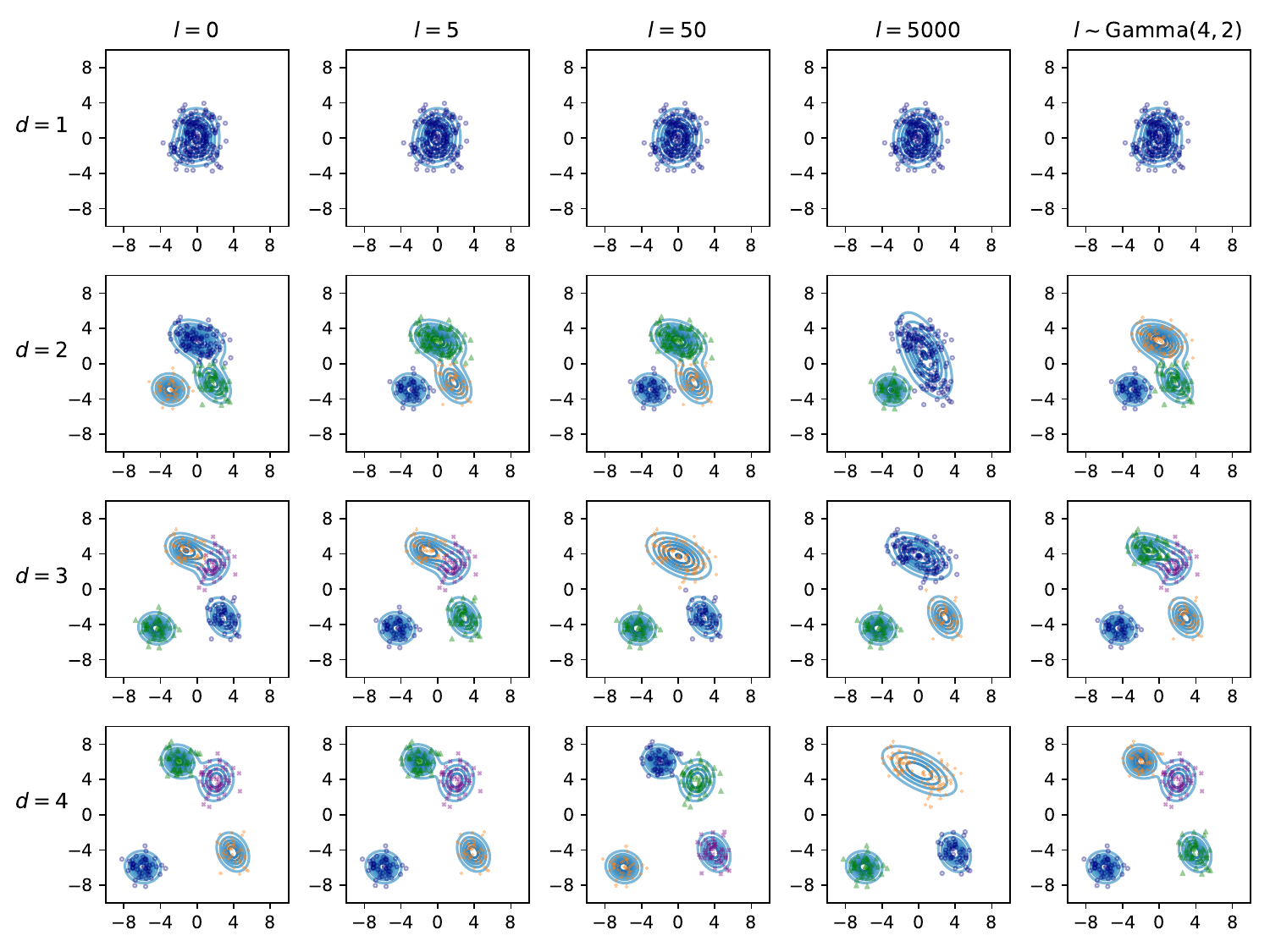}
	\caption{
		Squared-exponential MRMM
	}
\label{fig:syn-rbf-all}
\end{figure}

\begin{figure}
	\centering
	\includegraphics[width=\linewidth]{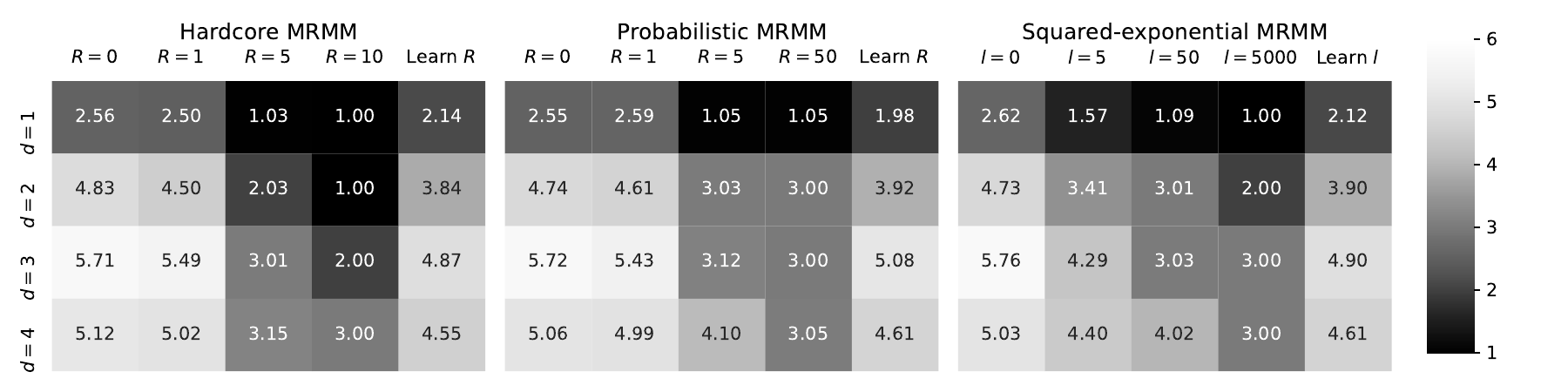}
	\caption{
		Posterior mean of the number of clusters $\EC$. 
	}
	\label{fig:syn-sum-EC}
\end{figure}

\begin{figure}
	\centering
	\includegraphics[width=\linewidth]{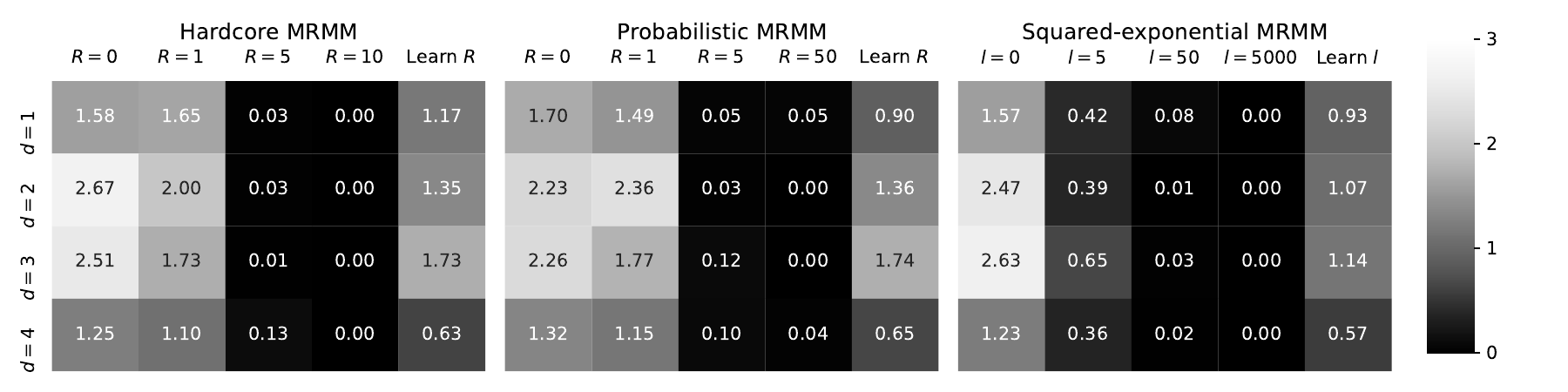}
	\caption{
		Posterior variance of the number of clusters $\VarC$.
	}
	\label{fig:syn-sum-VarC}
\end{figure}

\begin{figure}
	\centering
	\includegraphics[width=\linewidth]{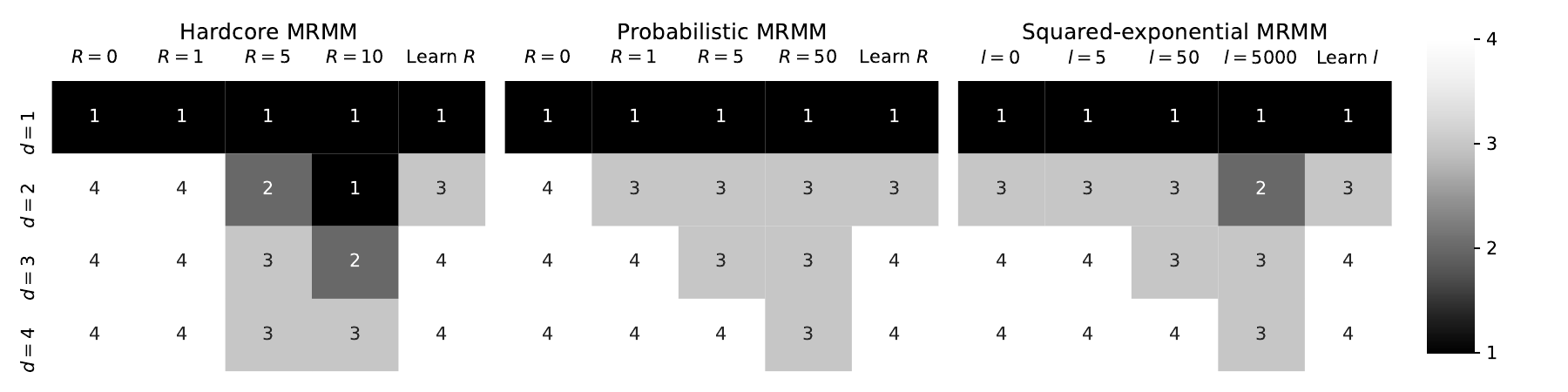}
	\caption{
		The number of clusters estimated from minimizing the posterior expectation of Binder’s loss function under equal misclassification costs, $\hC$.
	}
	\label{fig:syn-sum-hC}
\end{figure}

\begin{figure}
	\centering
	\includegraphics[width=\linewidth]{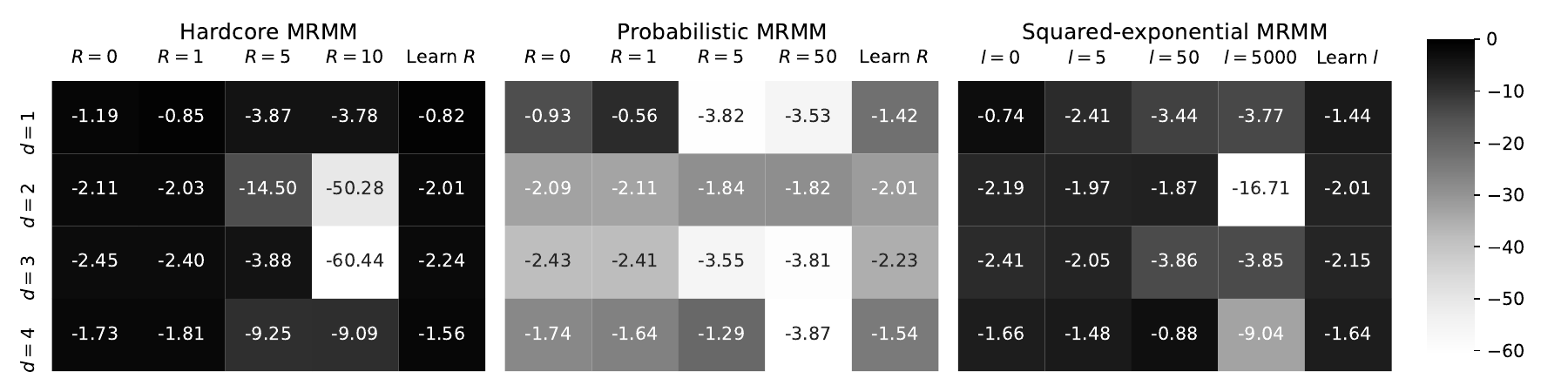}
	\caption{
		The difference between posterior testing likelihood and the testing likelihood under the ground truth model $M_0$, i.e. $\lnpTest - \ln \pgiven{\XX_{\text{test}}}{M_0}$.
	}
	\label{fig:syn-sum-lnpTest-diff}
\end{figure}

\begin{figure}
\centering
\includegraphics[width=\linewidth]{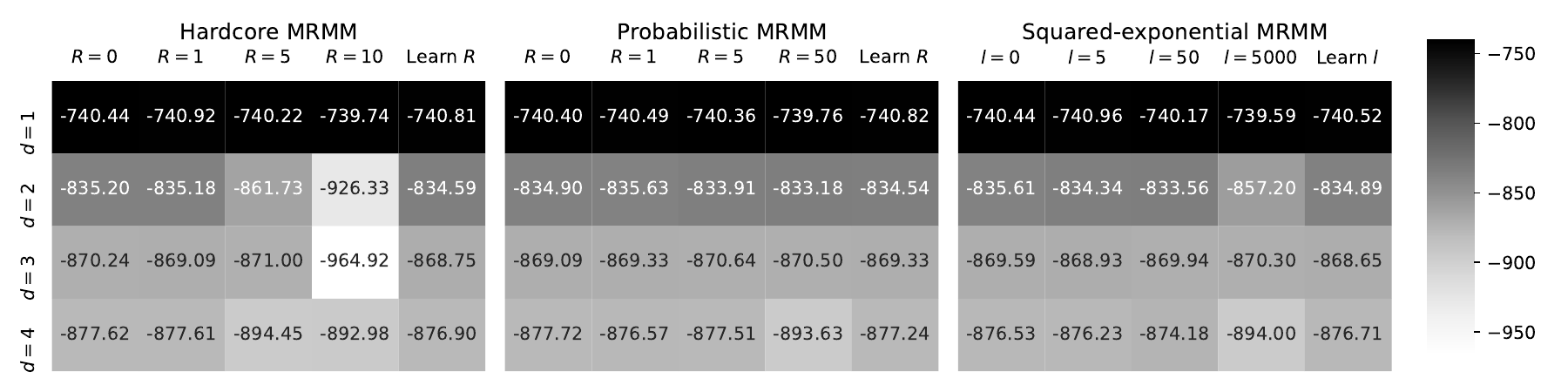}
\caption{
	The estimated log pseudo-marginal likelihood (LPML). 
}
\label{fig:syn-sum-LPML}
\end{figure}

\subsubsection{Setting thinning parameters via empirical Bayes}
\label{sec:eb2}

We consider draws from the following mixture of two components:
\\
\renewcommand{\arraystretch}{0.6}
\setlength{\arraycolsep}{3pt}
$
 \hspace*{1.6in}   y_{1},...,y_{n} \stackrel{iid}{\sim} 0.25N(-1,1) + 0.75SN(1,1,2),
$\\
\renewcommand{\arraystretch}{1}
\setlength{\arraycolsep}{5pt}
where $N(\mu,\sigma^{2})$ denotes the density of an univariate normal distribution with mean $\mu$ and variance $\sigma^{2}$. In addition, $SN(\xi,\omega,\alpha)$ denotes the density of an univariate skew normal distribution with location parameter $\xi$, scale parameter $\omega$, and shape parameter $\alpha$. 

We simulated a training dataset of size 600 and a test data with 300 observations were simulated independently from this. We model the dataset as a MRMM of univariate Gaussian kernel, with prior $p_{\Theta}(\theta)$ a Gaussian with mean zero and standard deviation 10 on the component locations, and an inverse-Gamma(1,1) prior on the variance of each mixture component.

The top-middle panel of Figure \ref{fig_EB_results} presents the kernel density estimate of pairwise distances. We observe that the distribution is unimodal with no clear local minimum. The estimated thinning radius, $\hat{\eta} = 7.342$, is notably large compared to the true between-cluster distance. The top-right panel of Figure \ref{fig_EB_results} displays the minimum of pairwise distance between cluster centers, $d_{\min,c}$. We observe the sharpest decrease in $d_{\min,c}$ occurs from $c=2$ to $c=3$, implying the emergence of a redundant cluster when fitting a mixture model with $c\geq 3$. The estimated thinning radius in this case is $\hat{\eta} = 1.960$, which is lower than the true between-cluster distance.

\begin{figure}[tbh!]
\begin{minipage}{0.6\textwidth}
	\centering
	\includegraphics[width=\textwidth]{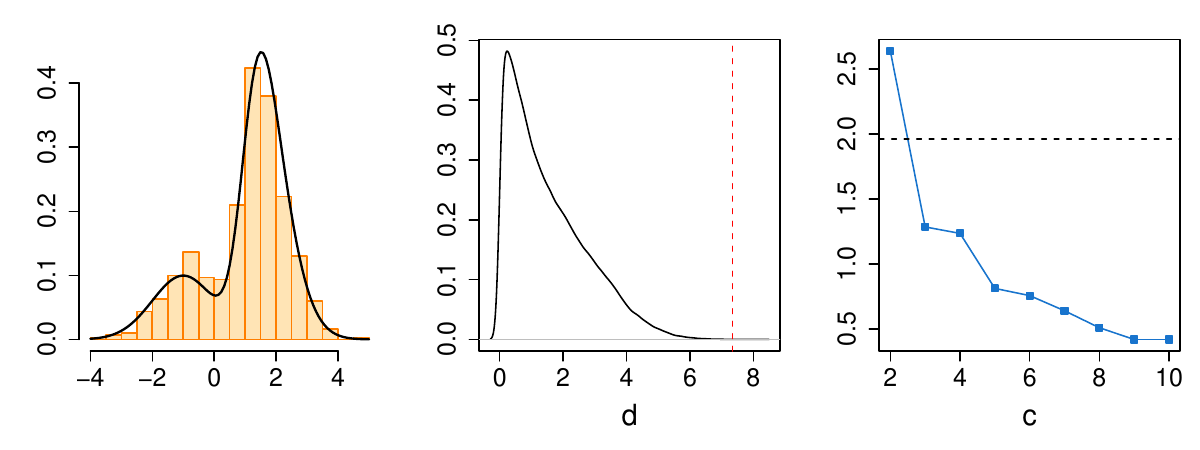}
    \includegraphics[width=\textwidth]{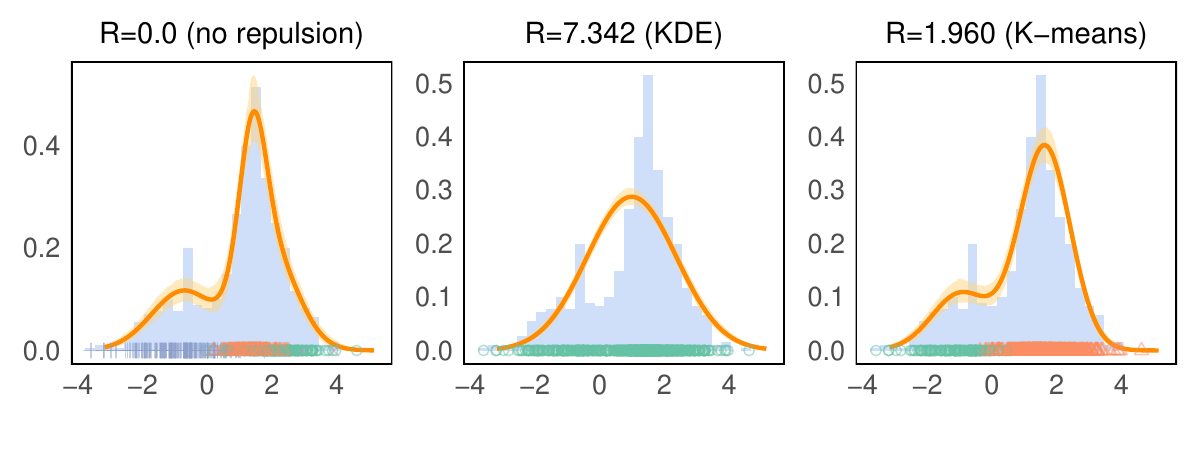}
    \end{minipage}
\begin{minipage}{0.38\textwidth}
	\caption{
	Top: (Left) Scatterplot of data with true mixture density. (Middle) Kernel density estimate of pairwise distances (Right)  $d_{\min,k}$ versus $k$.
    Bottom: Contour plot and cluster assignments of the univariate data for hardcore MRMM.}
	\label{fig_EB_results}
\end{minipage}
\end{figure}

The results for the hardcore MRMM are presented in Table \ref{table_hardcore} and in bottom panel of Figure \ref{fig_EB_results}. The clustering obtained using the approach of \cite{beraha2022mcmc} fails to identify two mixture components, collapsing all observations into a single cluster. In contrast, the $k$-means-based approach induces a moderate level of repulsion, so that it closely approximates the true number of clusters while not too much sacrificing the predictive performance compared to the non-repulsion scenario.

\begin{table}[!h]
\centering
\begin{tabular}{c|ccc|cc}
\hline
 Repulsion strength & $\EC$    & $\VarC$ & $\hC$ & $\lnpTest$ &	LPML\\
\hline\hline
\cellcolor{gray!0}{$R=0.0$ (no repulsion)} & \cellcolor{gray!0}{3.95} & \cellcolor{gray!0}{1.2311} & \cellcolor{gray!0}{4} & \cellcolor{gray!0}{-482.38} & \cellcolor{gray!0}{-984.73}\\\hline
$R=7.342$ (KDE-based) & 1.00 & 0.0000 & 1 & -529.08 & -1048.62\\\hline
\cellcolor{gray!0}{$R=1.960$ (K-means-based)} & \cellcolor{gray!0}{2.12} & \cellcolor{gray!0}{0.1220} & \cellcolor{gray!0}{2} & \cellcolor{gray!0}{-483.90} & \cellcolor{gray!0}{-993.93}\\
\hline
\end{tabular}
\caption{Posterior summaries of hardcore MRMM on the univariate dataset.}
\label{table_hardcore}
\end{table}

\subsection{Probabilistic MRMM on real datasets}
In this section, we report probabilistic MRMM results on real datasets. The model and parameter settings are the same with the hardcore MRMM reported in the paper, except for the thinning probability, which is fixed to 0.95 in all experiments below.  

\paragraph{Chicago 2019 homicide data} Results are shown in \Cref{fig:crime-soft} and \Cref{tbl:crime-soft}.

\begin{figure}[H]
	\includegraphics[width=\linewidth]{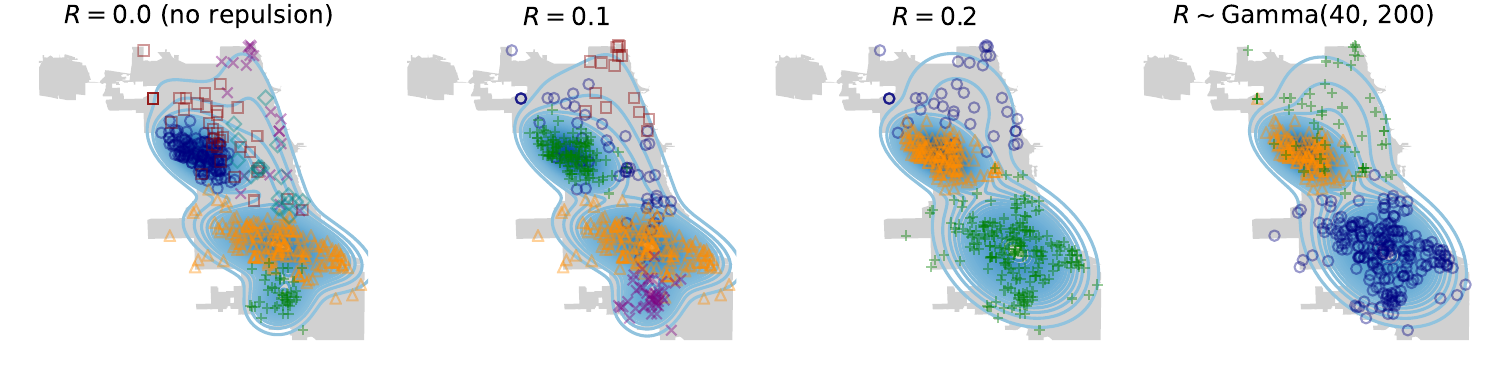}
	\caption{
		Contour plot and clustering of Chicago crime data from probabilistic MRMM. }
	\label{fig:crime-soft}
\end{figure}

\begin{table*}
	\centering
	\begin{tabular}{c|ccc|cc}
		\hline
		Repulsion strength      & $\EC$    & $\VarC$ & $\hC$ & $\lnpTest$ &	LPML	 \\ \hline\hline
		$R=0.0$ (no repulsion) &    5.26 &    0.3914 &       6 &   252.20 &  1351.28 \\\hline
		$R=0.1$ &    4.39 &    0.2703 &       5 &   251.21 &  1341.60 \\\hline
		$R=0.2$ &    3.00 &    0.0000 &       3 &   247.43 &  1321.11 \\\hline
		$R\sim\dGamma{40, 200}$ &    3.00 &    0.0040 &       3 &   246.73 &  1324.53 \\\hline
	\end{tabular}
	\caption{
		Posterior summaries of probabilistic MRMM on Chicago crime dataset. 
		Inferring the thinning radius yields the posterior mean and variance $\EE\sgiven R\XX = 0.15$, $\Var\rgiven R{\XX} = 0.0001$.}
	\label{tbl:crime-soft}
\end{table*}

\paragraph{Protein structural data} Results are shown in \Cref{fig:protein-soft} and \Cref{tbl:protein-soft}.

\begin{figure}[H]
	\includegraphics[width=\linewidth]{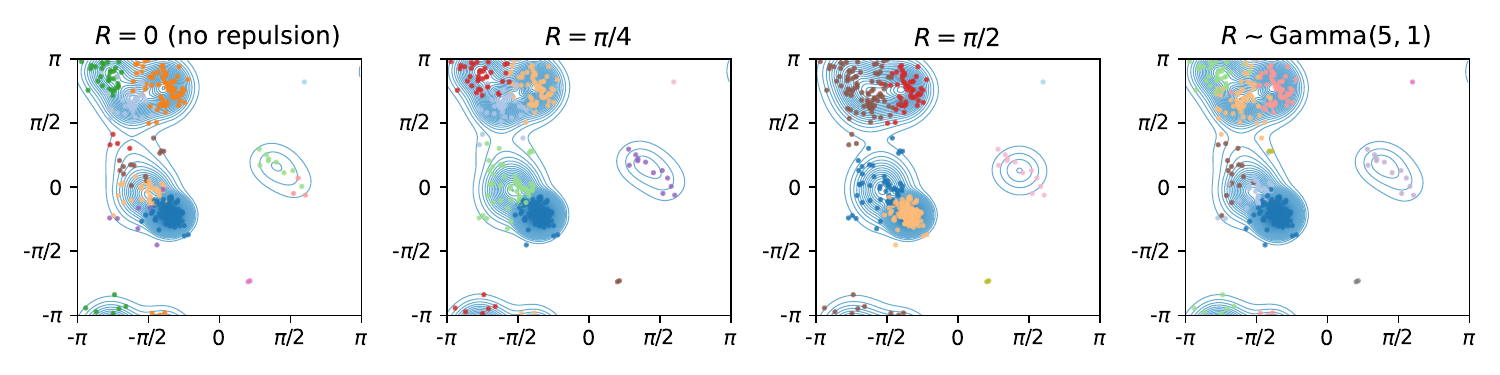}
	\caption{
		Contour plot and clustering of the protein data from probabilistic MRMM.}
	\label{fig:protein-soft}
\end{figure}

\begin{table*}
	\centering
	\begin{tabular}{c|ccc|cc}
		\hline
		Repulsion strength      & $\EC$    & $\VarC$ & $\hC$ & $\lnpTest$ &	LPML	 \\ \hline\hline
		$R=0$ (no repulsion) &   12.15 &    3.5369 &      13 &  -142.15 & -610.94 \\\hline
		$R=\pi/4$ &   10.14 &    1.5298 &       9 &  -142.36 & -618.13 \\\hline
		$R=\pi/2$ &    7.37 &    0.4056 &       7 &  -145.14 & -625.13 \\\hline
		$R\sim\dGamma{5, 1}$ &   10.85 &    1.9969 &      11 &  -143.64 & -622.55 \\\hline
	\end{tabular}
	\caption{
		Posterior summaries of probabilistic MRMM on the protein dataset. 
		Inferring the thinning radius yields the posterior mean and variance $\EE\sgiven R\XX = 0.18\pi$, $\Var\rgiven R{\XX} = 0.0017\pi^2$.
	}
	\label{tbl:protein-soft}
\end{table*}

\paragraph{Comparison with \citet{xie2019bayesian} on the Old Faithful dataset} Results are shown in \Cref{fig:faithful-soft} and \Cref{tbl:faithful-soft}.

\begin{figure}[H]
	\includegraphics[width=\linewidth]{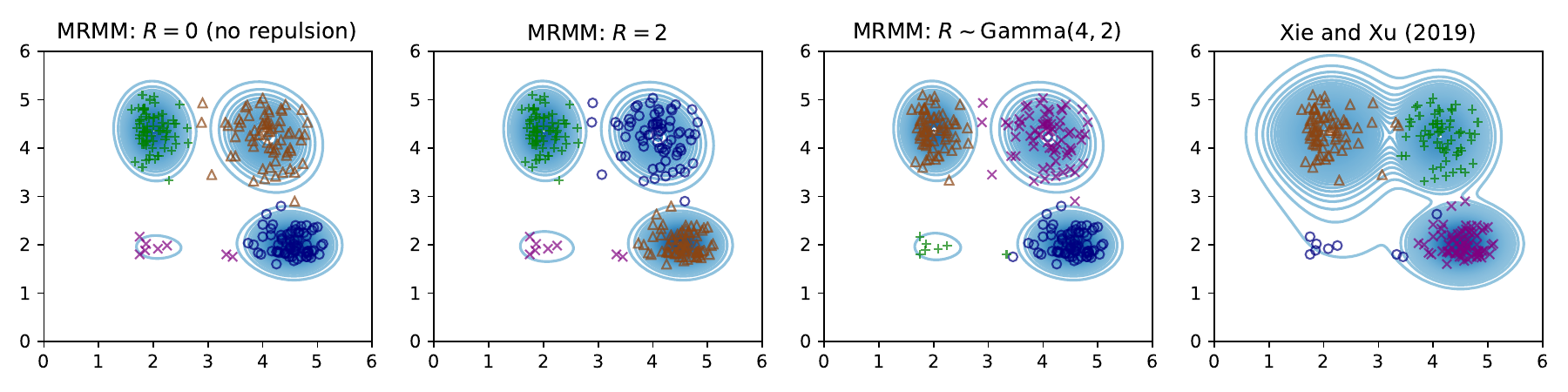}
	\caption{
		Contour plot and clustering of the Old Faithful geyser eruption data from probabilistic MRMM. }
	\label{fig:faithful-soft}
\end{figure}

\begin{table*}
	\centering
	\small
	\begin{tabular}{l|ccc|cc|c|c}
		\hline
		Model                & $\EC$ \hspace{-.2in}& $\VarC$ \hspace{-.2in} & $\hC$ & $\lnpTest$ \hspace{-.2in} &  LPML &  Runtime(s) & ESS/s   \\ \hline\hline
		\citet{xie2019bayesian} &    3.71 &    0.2116 &       4 &  -104.32 & -464.22 &  225.6 & 0.01 \\\hline
		MRMM &     &    &    &   &  &   &  \\\hline
		$R=0$ (no repulsion) &    4.02 &    0.0157 &       4 &   -95.83 & -420.53 &  257.8 & 14.31 \\\hline
		$R=2$ &    4.00 &    0.0000 &       4 &   -95.93 & -419.85 &  297.6 & 0.38 \\\hline
		$R\sim\dGamma{4, 2}$ &    4.01 &    0.0138 &       4 &   -95.96 & -420.94 &  287.0 & 2.66 \\\hline
	\end{tabular}
	\caption{
		Posterior summaries of probabilistic MRMM on the Old Faithful geyser eruption data. 
		Inferring the thinning radius yields the posterior mean and variance $\EE\sgiven R\XX = 1.39$, $\Var\rgiven R{\XX} = 0.1540$.}
	\label{tbl:faithful-soft}
\end{table*}

\paragraph{Comparison with \citet{bianchini2018determinantal} on the Galaxy dataset} Results are shown in \Cref{fig:galaxy-soft} and \Cref{tbl:galaxy-soft}.

\begin{figure}[H]
	\includegraphics[width=\linewidth]{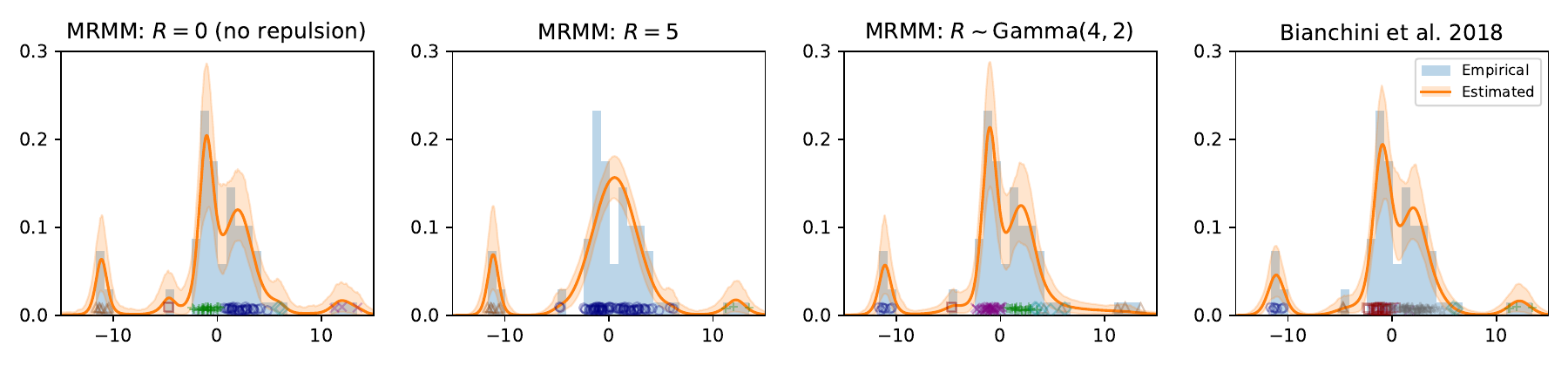}
	\caption{
		Contour plot and clustering of the Galaxy data from probabilistic MRMM.}
	\label{fig:galaxy-soft}
\end{figure}

\begin{table*}
	\centering
	\begin{tabular}{l|ccc|c|c|c}
		\hline
		Model & $\EC$  & $\VarC$  & $\hC$ &  LPML & Runtime (s) & ESS/s   \\ \hline\hline
		\citet{bianchini2018determinantal} &    6.00 &    1.2180 &       7 & -207.94 &   600.4 & 0.02 \\\hline
		MRMM: & & & & & & \\\hline
		$R=0$ (no repulsion) &    7.53 &    4.2370 &       6 & -209.66 &  734.4 & 46.9 \\\hline
		$R=5$ &    3.47 &    0.3772 &       3 & -212.36 &  410.4 & 172.4\\\hline
		$R\sim\dGamma{4, 2}$ &    6.23 &    1.8120 &       6 & -209.43 &  498.2 &13.0\\\hline
	\end{tabular}
	\caption{
		Posterior summaries of probabilistic MRMM on the Old Faithful geyser eruption data. 
		Inferring the thinning radius yields the posterior mean and variance $\EE\sgiven R\XX = 1.87$, $\Var\rgiven R\XX = 0.3228$.}
	\label{tbl:galaxy-soft}
\end{table*}

\end{document}